\documentclass[preprint,3p,11pt]{elsarticle}
\usepackage{amssymb}
\usepackage{epstopdf}
\usepackage{delarray}
\usepackage{multirow}
\usepackage{stmaryrd}
\usepackage{amsmath}
\usepackage{wrapfig}
\usepackage{subfig}
\usepackage{yfonts}
\usepackage{tikz}
\usepackage{delarray}
\usetikzlibrary{fadings,decorations,decorations.pathreplacing,patterns}
\usepackage{color}
\newcommand{\N}{\mathbb{N}}
\newcommand{\Z}{\mathbb{Z}}

\newtheorem{definition}{Definition}

\newtheorem{lemma}{Lemma}
\newtheorem{theorem}{Theorem}
\newtheorem{corollary}{Corollary}
\newtheorem{proposition}{Proposition}
\newdefinition{remark}{Remark}
\newdefinition{conjecture}{Conjecture}
\newproof{proof}{Proof}

\newcounter{theorem:D3}
\newcounter{corollary:decrease}
\newcounter{tmp}

\journal{TCS}
\renewcommand*{\today}{January 2013}

\begin{document}

\begin{frontmatter}

\title{ Kadanoff Sand Pile Model.\\
Avalanche  Structure and Wave Shape.\tnoteref{thx}}
\tnotetext[thx]{Partially supported by  IXXI (Complex System Institute, Lyon) and ANR projects Subtile and MODMAD.}
\author[ad1,ad2]{K\'evin Perrot}
\ead{kevin.perrot@ens-lyon.fr}
\author[ad1]{\'Eric R\'emila}
\ead{eric.remila@ens-lyon.fr}
\address[ad1]{Universit\'e de Lyon - Laboratoire de l'Informatique et du Parall\'elisme\\
(UMR 5668 - CNRS - ENS de Lyon - Universit\'e Lyon 1)\\
46 all\'e d'Italie 69364 Lyon Cedex 7, France}
\address[ad2]{Universit\'e Nice–Sophia Antipolis - Laboratoire I3S\\
(UMR 6070 CNRS)\\
2000 route des Lucioles, BP 121, F-06903 Sophia Antipolis Cedex, France}

\begin{abstract}
  Sand pile models are dynamical systems describing the evolution from $N$ stacked grains to a stable configuration. It uses local rules to depict grain moves and iterate it until reaching a fixed configuration from which no rule can be applied. Physicists L. Kadanoff {\em et al} inspire KSPM, extending the well known {\em Sand Pile Model} (SPM). In KSPM($D$), we start from a pile of $N$ stacked grains and apply the rule: $D\!-\!1$ grains can fall from column $i$ onto columns $i+1,i+2,\dots,i+D\!-\!1$ if the difference of height between columns $i$ and $i\!+\!1$ is greater or equal to $D$. Toward the study of fixed points (stable configurations on which no grain can move) obtained from $N$ stacked grains, we propose an iterative study of KSPM evolution consisting in the repeated addition of one grain on a heap of sand, triggering an avalanche at each iteration. We develop a formal background for the study of avalanches, resumed in a finite state word transducer, and explain how this transducer may be used to predict the form of fixed points. Further precise developments provide a plain formula for fixed points of KSPM(3), showing the emergence of a wavy shape.
\end{abstract}


\end{frontmatter}

%%%%%%%%%%%%%%%%%%%%%%%%%%%%%%%%%%%%%%%%%%%%%%%%%%%%%%%%
%%%%%%%%%%%%%%%%%%%%%%%%%%%%%%%%%%%%%%%%%%%%%%%%%%%%%%%%
%%
%%   Introduction
%%
%%%%%%%%%%%%%%%%%%%%%%%%%%%%%%%%%%%%%%%%%%%%%%%%%%%%%%%%
%%%%%%%%%%%%%%%%%%%%%%%%%%%%%%%%%%%%%%%%%%%%%%%%%%%%%%%%

\section{Introduction}

This paper is about cubic sand grains moving around on nicely packed columns in one dimension (the physical sand pile is two dimensional, but the support of sand columns is one dimensional). We follow the arbitrary convention that when sand grains can only fall in one direction according to iteration rules, this direction is the {\em right}. So when there is no ambiguity, we use any variation of the words {\em right} and {\em left} to refer to the direction of grain falls and its opposite.

\subsection{The framework}

Sand pile models were introduced in \cite{bak88} as systems presenting a critical self-organized behavior, a property of dynamical systems having critical points as attractors. In the scope of sand piles, starting from an initial configuration of $N$ stacked grains the local evolution of particles is described by one or more iteration rules. Successive applications of such rules alter the configuration until it reaches an attractor, namely a stable state from which no rule can be applied. {\em Self-organized criticality} (SOC) property means those attractors are critical in the sense that a small perturbation ---adding some sand grains--- involves an arbitrary deep reorganization of the system. Sand pile models were well studied in recent years.

Starting from $N$ stacked grains on column $0$ and no grain elsewhere, the first one dimensional {\em Sand Pile Model} (SPM) applies the rule: if the difference of height between columns $i$ and $i+1$ is greater or equal to 2, then one grain falls from column $i$ to column $i+1$. In \cite{goles93} the authors show that the set of reachable configurations endowed with the successor relation given by the iteration rule is a lattice. They furthermore provide a simple characterization of reachable configurations and give a plain formula describing the unique fixed point according to the number $N$ of grains. Finally they prove the convergence time to be in $O(n^{\frac{3}{2}})$. For a survey on SPM, see \cite{phan04}. A slight variant of SPM has been studied in \cite{durandlose98}: instead of applying the rule sequentially ---once at each iteration step---, the rule is applied in parallel on all possible positions. This model is called {\em Parallel SPM} (PSPM) and is deterministic. It reaches the same fixed point as SPM, but in time $O(n)$ (very smart and technical proof in \cite{durandlose98}). \cite{formenti07} and \cite{phan08} explore {\em Symmetric SPM} (SSPM) where grains can fall either to the right or to the left according to symmetric rules. The authors both give a simple characterization of fixed points shapes (there is no lattice structure anymore) in this sequential model. The {\em Parallel SSPM} (PSSPM), which is non deterministic since choices may occur on the top column, has been considered in \cite{formenti11} where a constructive proof shows that the fixed points shapes are the same as in SSPM, and in \cite{PSSPM} where the authors compare not only shapes, but also positions.

Another interesting question about SPM is the {\em prediction problem} (namely, the problem of computing the fixed point), which has been proved in \cite{moore99} to be in \textbf{NC}$^3$ for the one dimensional case (it means that the time needed to compute the fixed point is in $O(\log^3 N)$ on a parallel computer with a polynomial number of processors, where $N$ is the number of grains), and \textbf{P}-complete when the dimension is $\geq 3$.

  In \cite{kadanoff89}, Kadanoff proposed a generalization of classical models closer to physical behavior of sand piles, in which more than one grain can fall from a column during one iteration. Informally, Kadanoff sand pile model with parameter $D$, KSPM($D$), is a discrete dynamical system whose initial configuration is composed of a finite number $N$ of stacked grains, moving in discrete space and time according to an iteration rule: if the height difference between column $i$ and $i+1$ is greater or equal to $D$, then $D-1$ grains can fall from column $i$ to the $D-1$ adjacent columns on the right (see figure \ref{fig:rule}). Note that KSPM(2)=SPM.

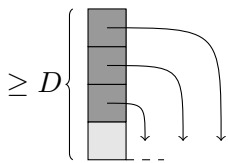
\begin{figure}[!h]
  \begin{center}
  \begin{tikzpicture}
  \foreach \y in {1,...,3}
    \filldraw[fill=black!40] (0,.5*\y) rectangle ++ (.5,.5);
  \filldraw[fill=black!10] (0,0) rectangle ++ (.5,.5);
  \draw[dashed] (.5,0) -- ++ (.5,0);
  \draw[decorate, decoration=brace] (-.2,0) -- node [left] {$\geq D$} ++ (0,2);
  \draw[->] (.25,.75) .. controls (.75,.75) .. (.75,.25);
  \draw[->] (.25,1.25) .. controls (1.25,1.25) .. (1.25,.25);
  \draw[->] (.25,1.75) .. controls (1.75,1.75) .. (1.75,.25);
\end{tikzpicture}
  \end{center}
  \caption{KSPM($D$) iteration rule.}
  \label{fig:rule}
\end{figure}

Figure \ref{fig:example} presents an example of evolution.

\begin{figure}[!h]
  \begin{center}
    %D=3,N=12
\begin{tikzpicture}[scale=.25]
  %h= 24
  \foreach \x/\h in {0/24}
    \foreach \y in {1,...,\h}
      \filldraw[fill=black!10] (\x,\y) rectangle ++ (1,1);
  %i=0
  \draw[->] (1.5,2) -- node[above]{$0$} ++ (1,0);
  %h= 22 1 1
  \foreach \x/\h in {0/22,1/1,2/1}
    \foreach \y in {1,...,\h}
      \filldraw[fill=black!10] (\x+3,\y) rectangle ++ (1,1);
  %i=0
  \draw[->] (6.5,2) -- node[above]{$0$} ++ (1,0);
  %h= 20 2 2
  \foreach \x/\h in {0/20,1/2,2/2}
    \foreach \y in {1,...,\h}
      \filldraw[fill=black!10] (\x+8,\y) rectangle ++ (1,1);
  %i=0
  \draw[->] (11.5,2) -- node[above]{$0$} ++ (1,0);
  %h= 18 3 3
  \foreach \x/\h in {0/18,1/3,2/3}
    \foreach \y in {1,...,\h}
      \filldraw[fill=black!10] (\x+13,\y) rectangle ++ (1,1);
  %i=0
  \draw[->] (16.5,2) -- node[above]{$0$} ++ (1,0);
  %h= 16 4 4
  \foreach \x/\h in {0/16,1/4,2/4}
    \foreach \y in {1,...,\h}
      \filldraw[fill=black!10] (\x+18,\y) rectangle ++ (1,1);
  %i=0
  \draw[->] (21.5,2) -- node[above]{$0$} ++ (1,0);
  %h= 14 5 5
  \foreach \x/\h in {0/14,1/5,2/5}
    \foreach \y in {1,...,\h}
      \filldraw[fill=black!10] (\x+23,\y) rectangle ++ (1,1);
  %i=0
  \draw[->] (26.5,2) -- node[above]{$0$} ++ (1,0);
  %h= 12 6 6
  \foreach \x/\h in {0/12,1/6,2/6}
    \foreach \y in {1,...,\h}
      \filldraw[fill=black!10] (\x+28,\y) rectangle ++ (1,1);
  %i=0
  \draw[->] (31.5,2) -- node[above]{$0$} ++ (1,0);
  %h= 10 7 7
  \foreach \x/\h in {0/10,1/7,2/7}
    \foreach \y in {1,...,\h}
      \filldraw[fill=black!10] (\x+33,\y) rectangle ++ (1,1);
  %i=0
  \draw[->] (36.5,2) -- node[above]{$0$} ++ (1,0);
  %h= 8 8 8
  \foreach \x/\h in {0/8,1/8,2/8}
    \foreach \y in {1,...,\h}
      \filldraw[fill=black!10] (\x+38,\y) rectangle ++ (1,1);
  %i=2
  \draw[->] (41.5,2) -- node[above]{$2$} ++ (1,0);
  %h= 8 8 6 1 1
  \foreach \x/\h in {0/8,1/8,2/6,3/1,4/1}
    \foreach \y in {1,...,\h}
      \filldraw[fill=black!10] (\x+43,\y) rectangle ++ (1,1);
  %i=2
  \draw[->] (48.5,2) -- node[above]{$2$} ++ (1,0);
  %h= 8 8 4 2 2
  \foreach \x/\h in {0/8,1/8,2/4,3/2,4/2}
    \foreach \y in {1,...,\h}
      \filldraw[fill=black!10] (\x+50,\y) rectangle ++ (1,1);
  %i=1
  \draw[->] (55.5,2) -- node[above]{$1$} ++ (1,0);
  %h= 8 6 5 3 2
  \foreach \x/\h in {0/8,1/6,2/5,3/3,4/2}
    \foreach \y in {1,...,\h}
      \filldraw[fill=black!10] (\x+57,\y) rectangle ++ (1,1);
\end{tikzpicture}
  \end{center}
  \caption{Example of an evolution from 24 stacked grains to the associated stable configuration. At each step, the arrow is labelled by the index of the column on which the rule has been applied.}
  \label{fig:example}
\end{figure}
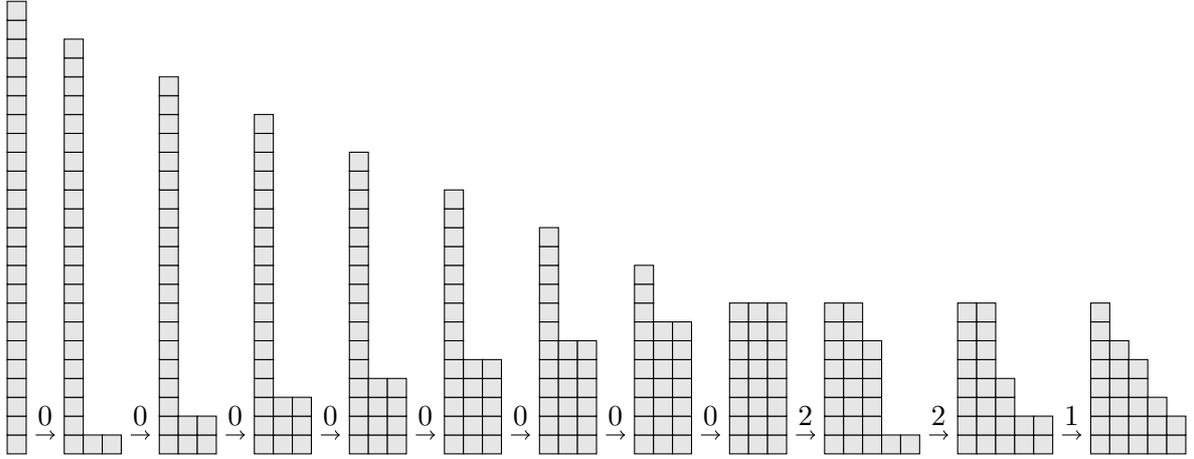

\subsection{Definitions and previous results}

More formally, sand pile models we consider are defined on the space of ultimately null decreasing integer sequences. Each integer represents a column of stacked sand grains and iteration rules describe how grains can move from columns. Let $h=(h_0,h_1,h_2,\dots)$ denote a {\em configuration} of the model, where each integer $h_i$ is the number of grains on column $i$. Configurations can also be given as sequences of {\em height differences} $\sigma=(\sigma_0,\sigma_1,\sigma_2,\dots)$, where for all $i \geq 0,~ \sigma_i=h_i-h_{i+1}$. We will use this latter representation throughout the paper, within the space of ultimately null non-negative integer sequences.

\begin{definition}
The   Kadanoff sand pile model with parameter $D$, KSPM($D$), is defined by:
  \begin{itemize}
    \item A set of \emph{configurations}, consisting in ultimately null non-negative integer sequences.
    \item A set of \emph{iteration rules}: we have a transition from a configuration $\sigma$ to a configuration $\sigma '$ on column $i$, and we note   $\sigma \overset{i}{\rightarrow} \sigma'$ when
    \begin{itemize}
\item $\sigma'_{i-1}=\sigma_{i-1} + D-1$ (for $i \neq 0$)
\item $\sigma'_i = \sigma_i - D$
\item $\sigma'_{i+D-1} = \sigma_{i+D-1} + 1$
\item $\sigma'_j = \sigma_j$ for $j \not\in  \{i-1, i, i+D-1 \}$. 
\end{itemize}
We also say that $i$ is {\em fired}.
  \end{itemize}
\end{definition}

Remark that according to the definition of the iteration rules, a condition for $\sigma'$ to be a configuration is that $\sigma_i \geq D$. We  note $\sigma \rightarrow \sigma'$ when there exists an integer $i$ such that $\sigma \overset{i}{\rightarrow} \sigma'$. The transitive closure of $\rightarrow$ is denoted by  $\overset{*}{\rightarrow}$, and we say that $\sigma'$ is {\em reachable} from $\sigma$ when $\sigma \overset{*}{\to} \sigma'$.

A basic property of the KSPM model is the \emph{diamond property}. If there exists two distinct integers $i$ and $j$ such that $\sigma \overset{i}{\rightarrow} \sigma'$ and $\sigma \overset{j}{\rightarrow} \sigma''$, then there exists a configuration $\sigma'''$ such that $\sigma' \overset{j}{\rightarrow} \sigma'''$  and $\sigma'' \overset{i}{\rightarrow} \sigma'''$. 

We say that a configuration $\sigma$ is \emph{stable}, or a \emph{fixed point} if no transition is possible from $\sigma$. 
As a  consequence of the diamond property and the termination of the evolution on finite configurations, one can easily check that, for each configuration $\sigma$, there exists a unique stable configuration, denoted by $\pi(\sigma)$, such that  $\sigma \overset{*}{\rightarrow} \pi(\sigma)$. Moreover,  for any configuration $\sigma '$ such that $\sigma \overset{*}{\rightarrow} \sigma '$, we have $\pi(\sigma') = \pi(\sigma)$ (see \cite{goles02} for details).

We are interested in the evolution from a finite number $N$ of grains to a stable configuration. Figure \ref{fig:lattice} depicts the set of reachable configurations for $D=3$ from the initial configuration with $24$ grains. Applying the rule once on the initial column leads for example to $(21,0,1,0,0,0,0,\dots)$ and we denote $0^\omega$ the infinite sequence of $0$, hence we write $(21,0,1,0^\omega)$ that configuration. For short, we state $\pi(N) = \pi((N,0^\omega))$.

\begin{figure}
  \begin{center}
    {\small
\begin{tikzpicture}[scale=1]
  \node (a) {$(24,0^\omega)$};
  \node (b) [below of=a,yshift=0cm] {$(21,0,1,0^\omega)$}
    edge [latex-] (a);
  \node (c) [below of=b,yshift=0cm] {$(18,0,2,0^\omega)$}
    edge [latex-] (b);
  \node (d) [below of=c,yshift=0cm] {$(15,0,3,0^\omega)$}
    edge [latex-] (c);
  \node (e) [below of=d,yshift=0cm] {$(12,0,4,0^\omega)$}
    edge [latex-] (d);
  \node (f) [below of=d,xshift=3cm,yshift=0cm] {$(15,2,0,0,1,0^\omega)$}
    edge [latex-] (d);
  \node (g) [below of=e,yshift=0cm] {$(9,0,5,0^\omega)$}
    edge [latex-] (e);
  \node (h) [below of=f,yshift=0cm] {$(12,2,1,0,1,0^\omega)$}
    edge [latex-] (e)
    edge [latex-] (f);
  \node (i) [below of=g,yshift=0cm] {$(6,0,6,0^\omega)$}
    edge [latex-] (g);
  \node (j) [below of=h,yshift=0cm] {$(9,2,2,0,1,0^\omega)$}
    edge [latex-] (g)
    edge [latex-] (h);
  \node (k) [below of=i,yshift=0cm] {$(3,0,7,0^\omega)$}
    edge [latex-] (i);
  \node (l) [below of=j,yshift=0cm] {$(6,2,3,0,1,0^\omega)$}
    edge [latex-] (i)
    edge [latex-] (j);
  \node (m) [below of=k,yshift=0cm] {$(0,0,8,0^\omega)$}
    edge [latex-] (k);
  \node (n) [below of=l,yshift=0cm] {$(3,2,4,0,1,0^\omega)$}
    edge [latex-] (k)
    edge [latex-] (l);
  \node (o) [below of=l,xshift=3cm,yshift=0cm] {$(6,4,0,0,2,0^\omega)$}
    edge [latex-] (l);
  \node (p) [below of=m,yshift=0cm] {$(0,2,5,0,1,0^\omega)$}
    edge [latex-] (m)
    edge [latex-] (n);
  \node (q) [below of=n,yshift=0cm] {$(3,4,1,0,2,0^\omega)$}
    edge [latex-] (n)
    edge [latex-] (o);
  \node (r) [below of=o,yshift=0cm] {$(8,1,0,1,2,0^\omega)$}
    edge [latex-] (o);
  \node (s) [below of=p,yshift=0cm,xshift=1.5cm] {$(0,4,2,0,2,0^\omega)$}
    edge [latex-] (p)
    edge [latex-] (q);
  \node (t) [below of=q,yshift=0cm,xshift=1.5cm] {$(5,1,1,1,2,0^\omega)$}
    edge [latex-] (q)
    edge [latex-] (r);
  \node (u) [below of=s,yshift=0cm,xshift=1.5cm] {$(2,1,2,1,2,0^\omega)$}
    edge [latex-] (s)
    edge [latex-] (t);
\end{tikzpicture}
}
  \end{center}
  \caption{The set of reachable configurations for $N=24$ and $D=3$. The initial configuration is on the top, and for short we denote $0^\omega$ the infinite sequences of 0.}
  \label{fig:lattice}
\end{figure}
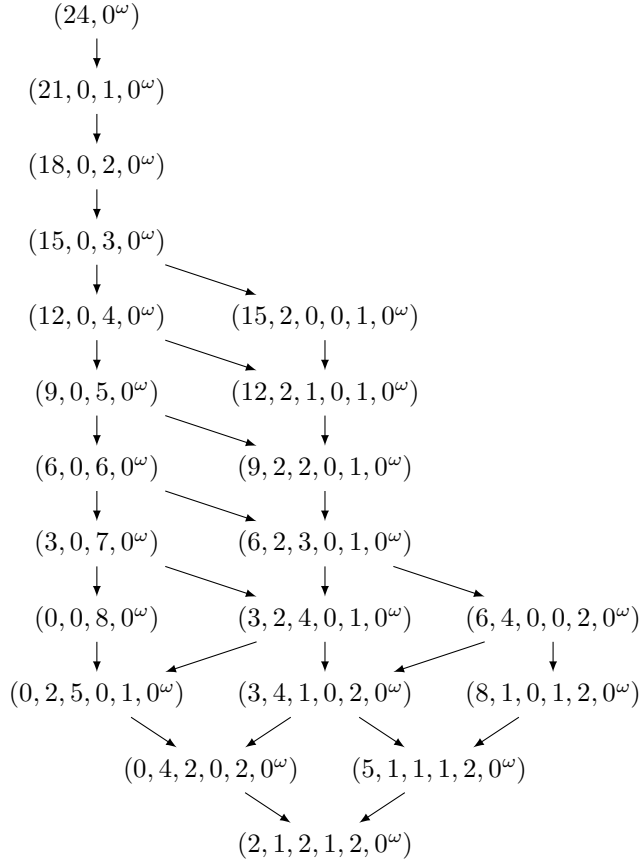

Sand pile models are specializations of {\em Chip Firing Games} (CFG). A CFG is played on a directed graph in which each vertex $v$ has a load $l(v)$ and a threshold $\theta(v)=deg^+(v)\footnote{$deg^+(v)$ denotes the out-degree of $v$.}$, and the iteration rule is: if $l(v)\geq \theta(v)$ then $v$ gives one unit to each of its neighbors (we say $v$ is fired). As a consequence, we inherit all  properties of CFGs. 

Kadanoff sand pile is referred to as a {\em linear chip firing game} in \cite{goles02}. The authors show that the set of reachable configurations endowed with the order induced by the successor relation has a lattice structure, in particular it has a unique {\em fixed point}. Since the model is non-deterministic, they also prove \emph{strong convergence} {\em i.e.}, the number of iterations to reach the fixed point is the same whatever the evolution strategy is. The morphism from KSPM(3) to CFG is depicted on figure \ref{fig:lcfg}.

When reasoning and writing formal developments about KSPM, it is much more convenient to think about its CFG representation because it is independent of the height of columns. Let us recall that throughout the paper, we consider sequences of height differences (except when explicitly specified) and the associated iteration rules where units of height difference move between columns.

  \begin{figure}[!h]
  \begin{center}
    \begin{tikzpicture}
  \node[circle, draw, fill=black!10] (n0) at (0,0) {$N$};
  \node[circle, draw, fill=black!10] (n-1) at (-1.5,0) {\scriptsize sink}
    edge [<-,out=10,in=170] (n0)
    edge [<-,out=-10,in=-170] (n0);
  \node[circle, draw, fill=black!10] (n1) at (1.5,0) {0}
    edge [->,out=170,in=10] (n0)
    edge [->,out=-170,in=-10] (n0);
  \node[circle, draw, fill=black!10] (n2) at (3,0) {0}
    edge [->,out=170,in=10] (n1)
    edge [->,out=-170,in=-10] (n1)
    edge [<-,out=-135,in=-45] (n0);
  \node[circle, draw, fill=black!10] (n3) at (4.5,0) {0}
    edge [->,out=170,in=10] (n2)
    edge [->,out=-170,in=-10] (n2)
    edge [<-,out=135,in=45] (n1);
  \node[circle, draw, fill=black!10] (n4) at (6,0) {0}
    edge [->,out=170,in=10] (n3)
    edge [->,out=-170,in=-10] (n3)
    edge [<-,out=-135,in=-45] (n2);
  \node[circle, draw, fill=black!10] (n5) at (7.5,0) {0}
    edge [->,out=170,in=10] (n4)
    edge [->,out=-170,in=-10] (n4)
    edge [<-,out=135,in=45] (n3);
  \node (n6) at (9,0) {}
    edge [->,dashed,out=170,in=10] (n5)
    edge [->,dashed,out=-170,in=-10] (n5)
    edge [<-,dashed,out=-135,in=-45] (n4);
  \node (n7) at (10.5,0) {}
    edge [<-,dashed,out=135,in=45] (n5);
\end{tikzpicture}
  \end{center}
  \caption{The initial configuration of KSPM(3) is presented as a CFG where each vertex corresponds to a column (except the sink, vertices from left to right corresponds to columns $0,1,2,3,\dots$) with a load equal to the difference of height between column $i$ and $i+1$. For example the vertex with load $N$ is the difference of height between column 0 ($N$ grains) and column 1 ($0$ grain).}
  \label{fig:lcfg}
\end{figure}
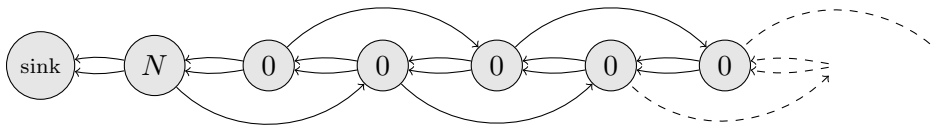

A recent study (\cite{goles10}) showed that in the two dimensional case the avalanche problem (given a configuration $\sigma$ and a column $i$ on which we add one grain, does it have an influence on index $j$?) on KSPM($D$) is \textbf{P}-complete, which points out an inherently sequential behavior.

\subsection{Our contribution}

The aim of this paper is to study the fixed points of the Kadanoff Sand Pile Model, and to describe them with a plain formula according to the parameter $D$ and the number of grains $N$. In section \ref{s:avalanches} we propose an inductive construction of the fixed point for a fixed parameter $D$ and $N$ grains according to the fixed point with $N-1$ grains: we add an $N^{th}$ grain on column $0$ of the fixed point with $N-1$ grains, which triggers an avalanche, ending when the fixed point with $N$ grains has been reached. An example of avalanche is given on figure \ref{fig:avalanche25}. We study the process of avalanches, and its emergent regularities. The main result of this section is a precise description of avalanches (Theorem \ref{theorem:peak}) when a property called {\em density} is fulfilled. We also prove for $D=3$ that on the right of a column $n$ in $O(\log N)$, the $N$ first avalanches have this property (Proposition \ref{lemma:meta2}), therefore we understand the behavior of the $N$ first avalanches for $D=3$, starting from column $n$.

From those regularities, we build words called {\em traces}, recording the way any grain has ever crossed a particular column $n$ during the $N$ first avalanches. Section \ref{s:transduction} outlines the construction of a finite state word transducer (Lemma \ref{lemma:transducer}), outputting the trace on column $n+(D-1)$ from the trace on column $n$. $i$ iterations of this transducer allows us to predict how any grain has crossed column $n+i(D-1)$, hence the behavior of avalanches from column $n+i(D-1)$ (see figure \ref{fig:intro}). From the trace on a column $n+i(D-1)$ for the $N$ first avalanches, we can straightforwardly compute the fixed point reached from $N$ stacked grains on the right of column $n+i(D-1)$. For $D=3$, we also prove that for a number of iterations $i$ in $O(\log N)$, the trace on column $n+i(D-1)$ is a prefix of $(ab)^\omega$ (Corollary \ref{corollary:decrease}). Consequently, regularities on traces induced by the transducer let us study the asymptotic form of fixed points on the right of column $n+i(D-1)$.

At the end of each section, we give a detailed study of the case $D=3$, ending in an asymptotic description of its fixed points according to the number $N$ of grains, pictured on figure \ref{fig:fp3}. Unfortunately, our complete study for the case $D=3$ does not trivially generalizes to any parameter $D$. Nevertheless, we hope that the present work already conveys interesting ideas on the handling of complex behaviors in discrete dynamical systems.

Some partial preliminary versions of the present work previously appeared in \cite{LATA} and \cite{MFCS}. 
 
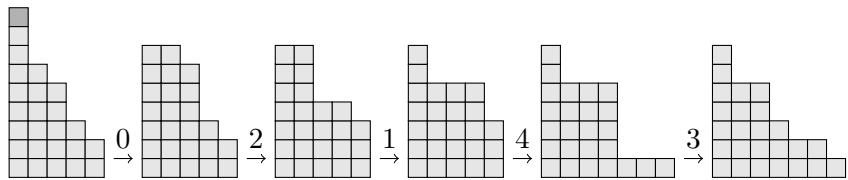
\begin{figure}[!h]
  \begin{center}
    \begin{tikzpicture}[scale=.25]
  %h= 9 6 5 3 2
  \foreach \x/\h in {0/8,1/6,2/5,3/3,4/2}
    \foreach \y in {1,...,\h}
      \filldraw[fill=black!10] (\x,\y) rectangle ++ (1,1);
  \filldraw[fill=black!30] (0,9) rectangle ++ (1,1);
  %i=0
  \draw[->] (5.5,2) -- node[above]{$0$} ++ (1,0);
  %h= 7 7 6 3 2
  \foreach \x/\h in {0/7,1/7,2/6,3/3,4/2}
    \foreach \y in {1,...,\h}
      \filldraw[fill=black!10] (\x+7,\y) rectangle ++ (1,1);
  %i=2
  \draw[->] (12.5,2) -- node[above]{$2$} ++ (1,0);
  %h= 7 7 4 4 3
  \foreach \x/\h in {0/7,1/7,2/4,3/4,4/3}
    \foreach \y in {1,...,\h}
      \filldraw[fill=black!10] (\x+14,\y) rectangle ++ (1,1);
  %i=1
  \draw[->] (19.5,2) -- node[above]{$1$} ++ (1,0);
  %h= 7 5 5 5 3
  \foreach \x/\h in {0/7,1/5,2/5,3/5,4/3}
    \foreach \y in {1,...,\h}
      \filldraw[fill=black!10] (\x+21,\y) rectangle ++ (1,1);
  %i=4
  \draw[->] (26.5,2) -- node[above]{$4$} ++ (1,0);
  %h= 7 5 5 5 1 1 1
  \foreach \x/\h in {0/7,1/5,2/5,3/5,4/1,5/1,6/1}
    \foreach \y in {1,...,\h}
      \filldraw[fill=black!10] (\x+28,\y) rectangle ++ (1,1);
  %i=3
  \draw[->] (35.5,2) -- node[above]{$3$} ++ (1,0);
  %h= 7 5 5 3 2 2 1
  \foreach \x/\h in {0/7,1/5,2/5,3/3,4/2,5/2,6/1}
    \foreach \y in {1,...,\h}
      \filldraw[fill=black!10] (\x+37,\y) rectangle ++ (1,1);
\end{tikzpicture}
  \end{center}
  \caption{An example of avalanche: starting from $\pi(24)$, we add one grain on column 0 (darkened on the leftmost configuration) and apply the iteration rule until reaching $\pi(25)$. At each step, the arrow is labelled by the index of the fired column.}
  \label{fig:avalanche25}
\end{figure}

\begin{figure}[!h]
  \centering
  \subfloat[Avalanches become regular very quickly. The regularity allows us to capture the way grains behave in a word $u$. $u$ describes completely how grains cross a particular column (the long vertical dashed line on the figure). $u$ barely consists in the concatenation of the relative positions of unstable columns, called {\em peaks}, for the $N$ first avalanches, within the $D-1$ columns preceding the vertical dashed line. Avalanches regularity also allows us, knowing the word $u$ for a column $i$ (the dashed vertical line), to compute the word $v$ describing how grains cross the column $i+D-1$ (where $D$ is the parameter of the model). The computation of $v$ is made from $u$ via a finite state word transducer \textswab T, which $n^{th}$ iteration computes the word describing how grains cross column $i+n(D-1)$.]{\label{fig:intro}\begin{tikzpicture}[scale=.75,baseline=0]
  \foreach \x in {0,...,4}{
    \fill[fill=black!10] (6-\x,0) -- ++ (-.5,0) -- ++ (0,.75+1.5*\x) -- ++ (.5,-.75) -- cycle;
    \fill[fill=black!30] (5.5-\x,0) -- ++ (-.5,0) -- ++ (0,1.5+1.5*\x) -- ++ (.5,-.75) -- cycle;
  }
  \foreach \x in {1,...,9}
    \draw[->] (1.2,8.5) .. controls (1+.5*\x,8.5) .. (1+.5*\x,7.75-.75*\x) node [fill=white,,xshift=1,yshift=12,inner sep=1pt] {\tiny \textswab T$^\x$};
  \draw[line width=1pt] (0,9) -- (0,0) -- ++ (6,0) -- cycle;
  \draw[dashed] (1,-.7) node [below] {\scriptsize ?} -- (1,11);
  \draw[dashed] (6,-.7) node [below] {\scriptsize $\Theta(\sqrt{N})$} -- (6,2);
  \draw[line width=1pt] (0,9.5) rectangle node[below,yshift=1] {\small $\downarrow$} ++ (.3,.3);
  \node at (-.3,10.6) {\small irregular};\node at (-.3,10.3) {\small avalanches};
  \node at (2.3,10.6) {\small regular};\node at (2.3,10.3) {\small avalanches};
  \node[fill=white,draw=black!30] at (1.5,2) {\scriptsize $N\!-\!1$ grains};
  \node[fill=white,inner sep=1pt] at (1,8.5) {\scriptsize $u$}; 
\end{tikzpicture}}
  \hspace{1cm}
  \subfloat[From the knowledge of the word $u$ describing how grains cross a column $i$, it is very easy to compute the configuration for every column $j \geq i$ (the concatenation of every words we obtain on a particular column describes how any single grain has ever cross that frontier). A precise study of the transducer for $D=3$ outlines, starting from any word, the very quick emergence of periodic words describing how grains cross columns as we iterate the transducer \textswab T. $O(\log N)$ iterations of \textswab T outputs periodic words of the form $ababab\dots$. Finally, The regularity of the words involves a regularity on the fixed points, which have asymptotically a completely wavy shape $212121\dots$.]{\label{fig:fp3}\begin{tikzpicture}[scale=.75,baseline=0]
  \draw[line width=1pt] (0,9) -- (0,0) -- ++ (6,0) arc (270:180:.5 and .75) arc (270:180:.5 and .75) arc (270:180:.5 and .75) arc (270:180:.5 and .75) arc (270:180:.5 and .75) arc (270:180:.5 and .75) arc (270:180:.5 and .75) arc (270:180:.5 and .75);
  \draw[densely dashed, line width=1pt] (0,9) -- ++ (2,-3);
  \draw[dashed] (1,-.7) node [below] {\scriptsize $O(\log N)$} -- (1,11);
  \draw[dashed] (2,-.2) node [below] {\scriptsize $O(\log N)$} -- (2,9.2);
  \draw[dashed] (6,-.7) node [below] {\scriptsize $\Theta(\sqrt{N})$} -- (6,2);
  \draw[->] (1.2,8.5) .. controls (1+.5*2,8.5) .. (1+.5*2,7.75-.75*2) node [fill=white,,xshift=4,yshift=12,inner sep=1pt] {\tiny \textswab T$^{\log N}$};
  \draw[line width=1pt] (0,9.5) rectangle node[below,yshift=1] {\small $\downarrow$} ++ (.3,.3);
  \node at (-.3,10.6) {\small irregular};\node at (-.3,10.3) {\small avalanches};
  \node at (2.3,10.6) {\small regular};\node at (2.3,10.3) {\small avalanches};
  \node[fill=white,draw=black!30] at (1.5,2) {\scriptsize $N\!-\!1$ grains};
  \node[fill=white,inner sep=1pt] at (1,8.5) {\scriptsize $u$}; 
\end{tikzpicture}}
  \caption{Presentation of the method (\ref{fig:intro}) and its application to KSPM(3) (\ref{fig:fp3}).}
\end{figure}

Describing and proving regularity properties, for models issued from basic dynamics is a present challenge for physicists, mathematicians, and computer scientists. There exists a lot of conjectures, issued from simulations, on  discrete dynamical systems with simple local rules (sandpile model \cite{dartois} or chip firing games, but also  rotor router  \cite{levine},  the famous Langton's ant \cite{gajardo}\cite{propp}...) but very few results have actually been proved.

%%%%%%%%%%%%%%%%%%%%%%%%%%%%%%%%%%%%%%%%%%%%%%%%%%%%%%%%
%%%%%%%%%%%%%%%%%%%%%%%%%%%%%%%%%%%%%%%%%%%%%%%%%%%%%%%%
%%
%%   Avalanches
%%
%%%%%%%%%%%%%%%%%%%%%%%%%%%%%%%%%%%%%%%%%%%%%%%%%%%%%%%%
%%%%%%%%%%%%%%%%%%%%%%%%%%%%%%%%%%%%%%%%%%%%%%%%%%%%%%%%

\section{Avalanches}\label{s:avalanches}

In order to study the fixed point for a parameter $D$ and $N$ grains, we will study the sequence of fixed points for a parameter $D$ and no grain, then $D$ and 1 grain, $D$ and 2 grains, $D$ and 3 grains, etc... Throughout the paper, the parameter $D$ is supposed to be fixed and can take any value greater or equal to 2, except when a particular value is specified. Computing the fixed point with $k$ grains, given the fixed point with $k-1$ grains, is very simple: we add one grain on column 0 and perform all the possible transitions. We call this process an {\em avalanche}.

We first explain why adding a grain on column 0 of the fixed point with $k-1$ leads to the fixed point with $k$ grains in subsection \ref{ss:inductive}, and then formally define an avalanche. If an avalanche verifies a certain property (if it is {\em dense}), then we show that its evolution is somehow “linear” (subsection \ref{ss:peaks}), and we give a precise description of its mechanism by the mean of particularly unstable columns named {\em peaks} (subsection \ref{ss:local}). The linear process description holds only when avalanches are dense, which happens intuitively and experimentally starting from a very small index, but we only managed to prove for $D=3$ that avalanches are asymptotically completely described as linear processes (subsection \ref{ss:avD=3}). \textit{Asymptotically completely} means that if we consider the avalanche occurring from the fixed point with $N-1$ grains to the fixed point with $N$ grains for a parameter $D$, denoting $size(D,N-1)$ the size of $\pi(N-1)$ (number of non empty columns) and $\mathcal L(D,N)$ the column index starting from which the avalanche is dense (the property is then true from column $\mathcal L(D,N)$ to column $size(D,N-1)-1$), then the ratio $\frac{\mathcal L(D,N)}{size(D,N-1)}$ tends to 0 as $N$ tends to $+\infty$.

%%%%%%%%%%%%%%%%%%%%%%%%%%%%%%%%%%
%
%   Inductive construction
%
%%%%%%%%%%%%%%%%%%%%%%%%%%%%%%%%%%

\subsection{Inductive construction}\label{ss:inductive}

In this paper, we are interested in the iterative process defined below. Starting with no grain, we successively add a single  grain on  column 0, and make all the possible firings until  a fixed point is reached. 

%We denote $\pi(k)$ the configuration obtained with this process after $k$ grain additions.

Let $\sigma$ be a configuration, $\sigma^{\downarrow 0}$ denotes the configuration obtained by adding one grain on column $0$. In other words, if $\sigma=(\sigma_0,\sigma_1,\dots)$, then $\sigma^{\downarrow 0}=(\sigma_0 +1 ,\sigma_1,\dots)$. 
Let $k >0$. Remark  that $\pi(k-1)^{\downarrow 0}$ is a reachable configuration from $(k,0^\omega)$, because we can use any sequence of firing from $ (k-1,0^\omega)$ to $\pi(k-1)$ to get an evolution from $(k,0^\omega)$ to $\pi(k-1)^{\downarrow 0}$. Thus,  with the uniqueness of the fixed point reachable from $(k,0^\omega)$,   we have   the recurrence formula (see figure \ref{fig:inductive}):
$$\pi(\pi(k-1)^{\downarrow 0})  =   \pi(k)$$

which, with the initial condition: $\pi(0) = 0^\omega$ allows an inductive computation of $\pi(k)$. 

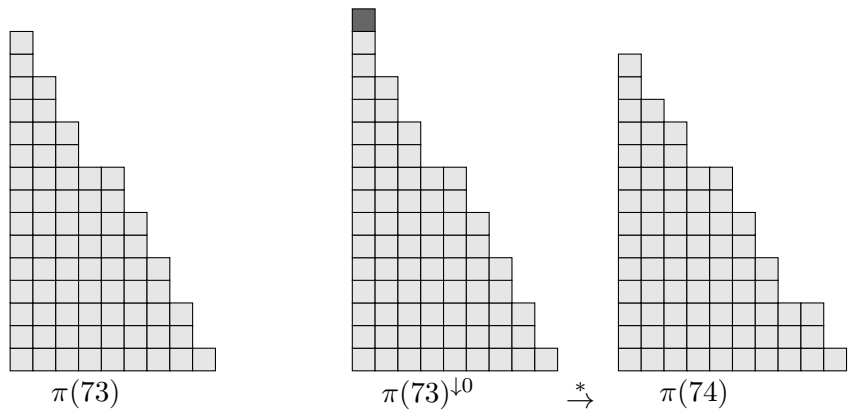
\begin{figure}[!h]
\begin{center}
  \begin{tikzpicture}
  %pi(73)
  \foreach \x/\h in {0/15,1/13,2/11,3/9,4/9,5/7,6/5,7/3,8/1}
    \foreach \y in {1,...,\h}
      \filldraw[fill=black!10] (\x*.3,\y*.3) rectangle ++ (.3,.3);
  %pi(73)^{\downarrow 0}
  \foreach \x/\h in {0/15,1/13,2/11,3/9,4/9,5/7,6/5,7/3,8/1}
    \foreach \y in {1,...,\h}
      \filldraw[fill=black!10] (4.5+\x*.3,\y*.3) rectangle ++ (.3,.3);
  \filldraw[fill=black!60] (4.5+0*.3,16*.3) rectangle ++ (.3,.3);
  %pi(74)
  \foreach \x/\h in {0/14,1/12,2/11,3/9,4/9,5/7,6/5,7/3,8/3,9/1}
    \foreach \y in {1,...,\h}
      \filldraw[fill=black!10] (8+\x*.3,\y*.3) rectangle ++ (.3,.3);
  %caption
  \node at (1,0) {$\pi(73)$};
  \node at (5.5,0) {$\pi(73)^{\downarrow 0}$};
  \node at (7.5,0) {$\overset{*}{\to}$};
  \node at (9,0) {$\pi(74)$};
\end{tikzpicture}
\end{center}
\caption{$D=3$, adding a grain on column 0 of $\pi(73)$ gives $\pi(73)^{\downarrow 0}$ and performing all the possible transitions until reaching a fixed point leads to $\pi(74)$.}
\label{fig:inductive}
\end{figure}

%%%%%%%%%%%%%%%%%%%%%%%%%%%%%%%%%%
%
%   Avalanches and peaks
%
%%%%%%%%%%%%%%%%%%%%%%%%%%%%%%%%%%

\subsection{Avalanches and peaks}\label{ss:peaks}

We are now interested in the description of the evolution from $\pi(k-1)^{\downarrow 0}$ to $\pi(k)$. However, from the non-determinacy of our model, this evolution is not unique. To overcome this issue we give below a formal definition of an evolution (a {\em strategy}) and distinguish a particular one from $\pi(k-1)^{\downarrow 0}$ to $\pi(k)$, which we think is the simplest, and define it as the {\em $k^{th}$ avalanche}.

A {\em strategy} is a sequence $s=(s_1,\dots,s_T)$. We say that $\sigma'$ is {\em reached} from $\sigma$ via $s$ when $\sigma \overset{s_1}{\rightarrow} \sigma'' \overset{s_2}{\rightarrow} \dots \overset{s_T}{\rightarrow} \sigma'$ and we note $\sigma \overset{s}{\rightarrow} \sigma'$. 
We also say,  for each integer $t$ such that $0 < t \leq T$, that the column $s_t$ \emph{is fired} at \emph{time} $t$ in $s$ (informally,  the index of the sequence is interpreted as time). 

 For any strategy $s$ and any nonnegative integer $i$, we state $|s|_i=\#\{ t | s_t = i \}$.  Let   $s^0$, $s^1$ be  two strategies such that $\sigma \overset{s^0}{\rightarrow} \sigma^0$ and $\sigma \overset{s^1}{\rightarrow} \sigma^1$.  We have the equivalence:  $[~\forall~ i, |s^0|_i = |s^1|_i~] \Leftrightarrow \sigma^0 = \sigma^1$.
 A strategy $s$ such that $\sigma \overset{s}{\rightarrow} \sigma'$ is called {\em leftmost} if it is the minimal strategy from $\sigma$ to $\sigma'$ according to lexicographic order. A leftmost strategy is such that at each iteration, the leftmost possible transition is performed. 

The {\em $k^{th}$ avalanche} $s^k$ is  the leftmost strategy from $\pi(k-1)^{\downarrow 0}$ to $\pi(k)$. Toward the study of fixed point configurations, we first consider the process of avalanches. Informally, we want to describe what happens when a new grain is added in a previously stabilized sand pile. For $D = 2$, {\em i.e.},  the classical SPM, this description is easy: the added grain moves rightwards until it reaches a stable position on two consecutive columns of same height. But, for $D > 2$,  the situation is not so simple.

The first Lemma is rudimentary but allows to simplify some notations for the rest of the paper.

\begin{proposition}\label{lemma:01}
  For each strategy $s$  such that  $\pi(N)^{\downarrow 0}  \overset{s} {\rightarrow}\pi(N+1)$ and any column $i \in \N$, we have  $|s|_i \in \{0,1\}$. 
\end{proposition}

\begin{proof}
  Let $s=(s_1,\dots,s_T)$ be a strategy such that  $\pi(N)^{\downarrow 0}  \overset{s} {\rightarrow}\pi(N+1)$. We have to  prove  that, for  
  $1 \leq l < m \leq T$, we have $s_l \not = s_m$ (obviously, $|s|_i \geq 0$ for all $i$). To do it,  we prove by induction on $t \leq T$ that for $1 \leq l < m \leq t$, we have $s_l \not = s_m$.
  
  For initialization this is obviously true for $t = 1$. 
   Now assume that the condition is satisfied for an integer $t$ such that $t < T$, and let $i$ be a column such that there exists an integer $l \leq t$ such that $i = s_l$. Let $\sigma$  be the  configuration such that $\pi(N)^{\downarrow 0} \overset{s_1}{\rightarrow} \dots \overset{s_t}{\rightarrow} \sigma$.
   
  Notice that the transitions which can possibly change the value of the current configuration at $i$ could be: $i$ (which decreases the value by $D$ units), $i+1$ (which increases the value by $D -1$ units) or  $i - D+1$ (which increases the value by $1$ unit).
  
  Thus we have $\sigma_i \leq  \pi(N)^{\downarrow 0}_i - D + D-1 +1$ since by definition,  between   $\pi(N)$ and $\sigma$, exactly one transition has occurred in   $i$, at most one transition has occurred in $i+1$, and at most one transition has occurred in $i - D+1$. For $i \geq 1$, we get $\sigma_i \leq  \pi(N)_i $. On the other hand, since $\pi(N)$ is a fixed point, we have:   $\pi(N)_i < D $,  which guarantees that $s_{t+1} \neq i$. For $i = 0 $, there is no possible transition in $i - D+1$, thus we get $\sigma_0 \leq  \pi(N)^{\downarrow 0}_0 - D + D-1 $, which is  $\sigma_0 \leq  \pi(N)_0 +1 - D + D-1$. Thus    $\sigma_0 \leq  \pi(N)_0 < D $ which also gives: $s_{t+1}  \neq 0$.
  
  This ensures that the result is true for $t+1$,  and,  by induction,  for $T$.\qed
\end{proof}

When talking about an avalanche $s$, Proposition \ref{lemma:01} allows us to write $i \in s$ and $i \notin s$ instead of $|s|_i=1$ and $|s|_i=0$ since no other value is possible. We denote by  $s_{[i, j]}$  the  subsequence of $s$ from $i$ to $j$ included.

We first explain the ``pseudo locality'' of avalanches: at a time $t+1$, a fired column can't be at distance greater ---neither on the left nor on the right--- than $D-1$ of the greatest fired column of $s_{[1,t]}$. Imagine, during an avalanche, that you follow the greatest fired column with a frame of size $2(D-1)-1$ centered on it, Lemma \ref{lemma:localdensity} tells that you won't miss any firing.

\begin{lemma}\label{lemma:localdensity}
Let $s^k=(s^k_1,\dots,s^k_T)$ be the $k^{th}$ avalanche. Let $r_t = \max s^k_{[1,t]}$.
\begin{itemize}
\item Assume that $s^k_{t+1} < r_t$.  Then $s^k_{t+1}$  is the largest column number satisfying this inequality, which has not  yet been fired at time $t$. In other words:
$$ s^k_{t+1} = \max \{i ~|~  i  < r_t \text{ and } i \notin s^k_{[1,t]} \}$$
 Moreover,  we have $r_t - s^k_{t+1} < D-1$.
\item Assume that $s^k_{t+1} > r_t$. Then we have  $s^k_{t+1} \leq r_t + D-1$.
\end{itemize}

\end{lemma}

\begin{proof}
We order fired columns by causality. Precisely, a column $i$ has  two potential predecessors, which are
$i +1$ and   $ i - D+1$. State $i = s^k_u$. These columns are predecessors of $i$ if they are elements of $s^k_{[1, u ]}$, i.e if they are fired before $i$. Using the transitive closure,  we define a partial order relation (denoted $<_{caus.}$) on fired columns for $s^k$.

Now,  consider the set $A_{t+1 }$ of ancestors of $s^k_{t+1 }$ ({\em i.e.}  the set of columns $i$ such that $i <_{caus.} s^k_{t+1}$)  and the set $S_t$ of columns which have $s^k_t$ as an ancestor ({\em i.e.} columns $i$ such that $s^k_t  <_{caus.} i$).
We necessarily have  $r_t \in A_{t+1 }$.   Otherwise, we have  $A_{t+1 } \cap S_t =  \emptyset$, and this  allows  another strategy $s'$, constructed from $s^k$ by postponing the transitions at  $r_t$ and elements of $S_t$ after the transition on  $s^k_{t+1}$. This contradicts the fact that $s^k$ is leftmost.

Let  $(i_0, i_1, \dots, i_p)$ be a finite sequence such that $i_0 = r_t$, $i_p = s^k_{t+1}$ and,  for each $j$ with $0 \leq j < p$, $i_j$ is a predecessor of $i_{j+1}$. Such a sequence exists since $r_t \in A_{t+1 }$. Let us prove by induction that $i_j = r_t -j$: 
this is true for $j = 0$. Assume it is true until the integer $j < p$. We have either $i_{j+1} = i_j -1$ or  $i_{j+1} = i_j +D -1$.
But from the induction hypothesis, $i_j +D -1$ is an ancestor of $i_{j+1}$, thus $i_{j+1} = i_j -1$.
%$r_t > i_{j+1}$ by definition (since $i_{j+1}$ is fired at time at most $t$), so if $i_{j+1}=i_j+D-1$, then $j \geq D-1$ and by induction hypothesis $i_{j+1}=r_t-j+D-1$ and has already been fired, thus $i_{j+1}=i_j-1$.
This gives that   $s^k_{t+1}$  is the largest column number $i$ such that $i < r_t$ and $i \notin s^k_{[1,t]}$. 

Now if we assume, by contradiction, that  $p \geq  D-1$, then  $r_t -D+1$ is not a predecessor of $r_t$, which yields that $r_t$ has no predecessor, which is a contradiction. This gives the  inequality   of the first item. 
The second item  is obvious, since $s^k_{t+1}$ ha a unique predecessor which is $s^k_{t+1}-D+1$.\qed
\end{proof}

Lemma \ref{lemma:localdensity} induces a partition of fired columns between those which make a progress ({\em i.e.} increases the greatest fired column) and those which do not. This distinction is important in further development, so let us give progress firings a name. Let $s^k=(s^k_1,\dots,s^k_T)$ be an avalanche, a column $s^k_{t}$ is called a {\em peak} if and only if $s^k_{t} > \max s^k_{[1,t-1]}$.

  Remark that two peaks $p\not =q$ can be compared using chronological ($<_T$) or spatial ($<_S$) orders. Nevertheless, by definition of peaks we obviously have $p <_T q \iff p <_S q$.

%%%%%%%%%%%%%%%%%%%%%%%%%%%%%%%%%%
%
%   Pseudo local process
%
%%%%%%%%%%%%%%%%%%%%%%%%%%%%%%%%%%

\subsection{Pseudo local process}\label{ss:local}

The next Lemma explains precisely the way peaks appear, as soon as a $D-1$ successive columns are fired. It follows an intuitive idea: a peak at time $t+1$ is a column which only receives grains from the left part of the sand pile (within $s_{[1,t]}$). Therefore, the amount it receives is at most 1 and a peak must have an initial value of $D-1$ units of height difference. Also, a non-peak column isn't fired when it receives 1 unit of height difference so it has to wait for its right neighbor to be fired, in a kind of chain reaction.

\begin{lemma}\label{lemma:D-1}
 Let $s^k$ be the $k^{th}$ avalanche. Assume that there exists a column  $l$ such that $\llbracket l ; l+D-2 \rrbracket \subseteq s^k$, and a fired column $i' \in s^k$ such that $i'\geq l+D-1$. Let  $l'$ be  the lowest peak such that  $l' \geq l+D-1$.\\
There exists a time $t$ such that 
  \begin{tabular}[t]\{{l}.
    if $i$ is such that $l' -D+1 < i \leq  l'$ then $i \in s^k_{[1,t]}$\\
    if $i$ is such that $l'< i$ then $i \notin s^k_{[1,t]}$
  \end{tabular}\\
 Moreover, let $\sigma^t$ denote the configuration obtained from $\pi(k-1)$ via $s^k_{[1,t]}$, then $l'$ is the lowest integer such that  $l' \geq l + D-1$ and $\sigma^t_{l'} = D-1$.
\end{lemma}

Informally, the lemma above claims that the space threshold $l'$ induces a corresponding time threshold  $t$: columns fired before time $t$ are on  the  left of $l'$,  while columns fired after  time $t$ are on  the  right  of $l'$.

\begin{proof}
 Let $t_0$ be the time when $s^k_{t_0} = l'$, i.e. the   first  time such that $s^k_{t_0} \geq l+D-1$,  and let $j$ be the largest integer such that,  for all $j' \in \llbracket 0; j \rrbracket$, we have $s^k_{t_0 +j'} = s^k_{t_0} -j'$.  Let us state $t = t_0 +j$. We have $j < D-1$.
 
 Let   $i$, with $ i < l'$,  such that $i \notin s^k_{[1,t]}$. We claim that we have:  $i \notin s^k$. To prove it, we prove by induction that for any $t' \geq t$,  $i \notin s^k_{[1,t']}$. Assume that this is satisfied for a fixed $t'$. This means that all the transitions of $s^k_{[t+1,t']} $  are done on columns larger than $l'$.  Thus, $ \sigma^{t'}_i =   \sigma^{t}_i$ and no transition is possible  on $i$ for $\sigma^{t}$ since  $s^k$ is leftmost (the only potential column to be fired is $ s^k_{t_0} -j -1$, but by assumption, either this column has been previously fired,  or it cannot be fired by definition of $j$, according to Lemma \ref{lemma:localdensity}) .

By contraposition,  it follows that for each column  $i \in \llbracket l;l+D-1 \llbracket$,  we have $i \in s^k_{[1,t]}$.  A   simple (reversed) induction shows that,   for $i \in \llbracket l+D-1;l' \rrbracket$   we have $i \in s^k$, since  by hypothesis $i+1$, and $i+1 -D$ both are in $s^k$.  Thus, by contraposition of the claim above,  for $i \in \llbracket l+D-1;l' \rrbracket$,   we have $i \in s^k_{[1,t]}$. 
This gives  the the fact that for all $i$ with $l' -D+1 < i \leq  l'$, $i \in s^k_{[1,t]}$

The fact that for all $i$ with $l'< i$, $i \notin s^k_{[1,t]}$ is trivial,  by definition of $t_0$ and $t$.

 We have $l' > l + D-2$ and $\sigma^t_{l'} = D-1$. assume that there exists $l''< l'$ satisfying the same properties. Notice that 
 for $t_0 \leq  t' \leq t$   we have $s^k_{t'} > l' - D-1$. Thus the time $t_1$ such that $s^k_{t_1} = l'' -D+1$  is such that $t_1 < t_0$. 
  That means that $ l'' $  should have been fired before $t_0$,  a contradiction.\qed
\end{proof}

Lemma \ref{lemma:D-1} describes in a very simple way the behavior of avalanches. Thanks to it, the study of an avalanche can be turned into a pseudo linear execution, in which transitions are organized in a clear fashion:

\begin{theorem}\label{theorem:peak}
  Let $s^k=(s^k_1,\dots,s^k_T)$ be the $k^{th}$ avalanche and $(p_1,\dots,p_q)$ be its sequence of peaks. Assume that  there exists a column $l$,  such that for each column  $i$ with $i \in \llbracket l ; l+D-1 \llbracket$, $i \in s^k$. 
  Then for any column $p$ such that $p \geq l+D-1$, 
  $$p\text{ is a peak of }  s^k  \iff \pi(k-1)_p = D-1 \text{ and } \exists j \text{ s.t. } p_j < p \leq p_j+D-1.$$
  
  Furthermore, let $j \leq q$ and $t$ such that $p_j=s^k_t$, with $p_j \geq l+D-1$. Then
  $$T \geq t+p_j-p_{j-1}-1 \text{ and for all } t' \text{ s.t. } t < t' \leq t+p_j-p_{j-1}-1,~ s^k_{t'}=s^k_{t'-1}-1.$$

 \end{theorem}

 A graphical representation of this statement is given on figure \ref{fig:peak}.

\begin{figure}[!h]
\begin{tikzpicture}
  \filldraw[fill=black] (3*.5,0) rectangle ++ (.5,.5);
  \foreach \x in {7,8,11,16,23}
    \filldraw[fill=black!30] (\x*.5,0) rectangle ++ (.5,.5);
  \foreach \x in {0,...,25}
    \draw (\x*.5,0) rectangle ++ (.5,.5);
  \draw[densely dashed] (13,0) -- ++ (.5,0);
  \draw[densely dashed] (13,.5) -- ++ (.5,0);
  \draw[densely dashed] (0,0) -- ++ (-.5,0);
  \draw[densely dashed] (0,.5) -- ++ (-.5,0);
  \draw[line width=2pt] (0,0) -- ++ (2.5,0) -- ++ (0,.5) -- ++ (-2.5,0) -- cycle;
  \node at (.25,-.25) {$l$};
  \draw[-latex] (3*.5+.25,1) -- ++ (0,-.5);
  \draw[-latex] (3*.5+.25,.5) parabola[bend pos=0.5] bend +(0,.5) +(4*.5,0);
  \draw[-latex] (4*.5+.25,.5) parabola[bend pos=0.5] bend +(0,.5) +(4*.5,0);
  \draw[-latex] (8*.5+.25,.5) parabola[bend pos=0.5] bend +(0,.5) +(3*.5,0);
  \draw[-latex] (9*.5+.25,.5) parabola[bend pos=0.5] bend +(0,.5) +(7*.5,0);
  \foreach \x in {5,6,7,10,11,13,14,15,16}
    \draw[-latex] (\x*.5+.25,0) parabola[bend pos=0.5] bend +(0,-.3) +(-.5,0);
\end{tikzpicture}
\caption{Illustration of Theorem \ref{theorem:peak} with $D=6$; surrounded columns $l$ to $l+D-2$ are supposed to be fired; black column is the greatest peak strictly lower than $l+D-1$; a column is grey if and only if its value is $D-1$; following arrows depicts the avalanche}
\label{fig:peak}
\end{figure}
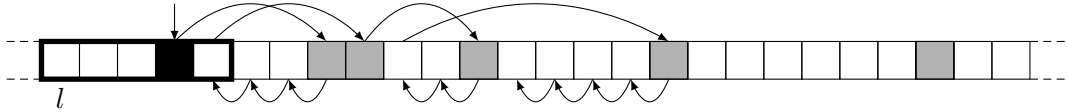

\begin{proof}
  The first part  is a straight induction on Lemma \ref{lemma:D-1}.\\
  The second part follows from an induction summed up in the following fact: any column $i$ such that $\pi(k-1)_i < D-1$ must wait for its right neighbor $i+1$ to be fired, and it should be fired when both $i+1$ and $i-D+1$ have been fired (besides, $i-D+1$ has already been fired). Since any of such $i$ is fired to reach a fixed point, we have for all peak $p_j \geq l + D-1$ that $T \geq t+p_j - p_{j-1} -1$.\qed
\end{proof}

Theorem \ref{theorem:peak} easily allows us to compute the right part of the $k^{th}$ avalanche (from column $l + D-1$), only knowing $\pi(k-1)$. The sequence of peaks is computed as follows. The first one is the lowest column $i$ greater or equal to $l+D-1$ such that $\pi(k-1)_i=D-1$. Then, given a peak $i$, the next one is the lowest  $j$ such that $\pi(k-1)_j = D-1$  and $j- i \leq D-1$. If such a $j$ does not exist, then there is no more peak and $i$ is the largest fired column.

We can distinguish two movements within an avalanche: before a certain column $l$ it has an unknown behavior, and from $l$ to the end the behavior is pseudo local: when an index is fired ahead (on the right) then any `hole' is filled before the progress can continue. Column $l$ depends on parameter $D$ and the number of grains $N$. We say that the $k^{th}$ avalanche $s^k$ is {\em dense from $l$ to $m$} when $m$ is the greatest fired column ($\forall~ i>m,~ i \notin s^k$) and any column between $l$ and $m$ included has been fired ($\forall~ i \in \llbracket l;m \rrbracket,~ i \in s^k$). A consequence of Theorem \ref{theorem:peak} is that the avalanche $s^k$ considered is dense starting at $l$,  where $l$ denotes the parameter in the statement of the Theorem. We will define the {\em global density column} $\mathcal L(D,N)$  as the minimal column such that for any avalanche $s^k$, with $k \leq N$, $s^k$ is dense from $\mathcal L(D,N)$. When parameters $D$ and $N$ are fixed, we sometimes simply denote $\mathcal L$. The formal definition of the global density column is:

\begin{definition}
  $\mathcal L'(D,k)$ is the minimal column such that the $k^{th}$ avalanche is dense starting at $\mathcal L'(D,k)$:
  $$\mathcal L'(D,k)=\min \{ l \in \N ~|~ \exists m \in \N \text{ such that } \forall~ i \in \llbracket l ; m \rrbracket, i \in s^k \text{ and } \forall~ i > m, i \notin s^k \}.$$
  Then, the global density column $\mathcal L(D,N)$ is defined as:
  $$\mathcal L(D,N)=\max \{\mathcal L'(D,k) ~|~ k \leq N \}.$$ 
\end{definition}  
  
See figure \ref{fig:gdc-3-0-100} for an illustration of the global density column $\mathcal L(3,N)$.
  
An important direct implication of Theorem \ref{theorem:peak} is that if there exists a column $l$ such that for the $k^{th}$ avalanche $s^k$, we have for all $i \in \llbracket l ; l+D-1 \llbracket$, $i \in s^k$, then for all $j \in \llbracket l+D-1 ; \max s^k \llbracket$, we have $j-(D-1), j, j+1 \in s^k$. From the definition of the model, those three firings add respectively $1,-D$ and $D-1$ units of height difference to $\pi(k)_j$. Since no other firing of the $k^{th}$ avalanche affects $\pi(k)_j$, we have $\pi(k)_j=\pi(k-1)_j$. This behavior accounts for an important part of the fixed point, and this part keeps the same shape (see an example on figure \ref{fig:6-1068-1069}). But if the shape is the same, then it intuitively hints some similarity between successive avalanches.

\begin{figure}[!h]
\begin{tikzpicture}
  %1068
  \node at (-1,.25) {\small $\pi(1068)$};
  \foreach \x in {10,...,21}
    \fill[fill=black!10] (\x*.5,0) rectangle ++ (.5,.5);
  \foreach \x in {13,14,16,21}
    \fill[fill=black!30] (\x*.5,0) rectangle ++ (.5,.5);
  \foreach \x/\n in {0/5,1/0,2/0,3/4,4/1,5/5,6/5,7/4,8/0,9/5,10/2,11/0,12/3,13/5,14/5,15/4,16/5,17/4,18/3,19/2,20/1,21/5,22/4,23/3,24/2,25/1}
    \draw (\x*.5,0) rectangle node {\scriptsize \n} ++ (.5,.5);
  \draw[line width=2pt, densely dashed] (10*.5,.6) -- ++ (0,-1.7) node [below] {\small $\mathcal L(6,1069)$};
  %1069
  \node at (-1,-.75) {\small $\pi(1069)$};
  \foreach \x/\n in {0/0,1/0,2/5,3/3,4/0,5/5,6/4,7/3,8/0,9/5,10/2,11/0,12/3,13/5,14/5,15/4,16/5,17/4,18/3,19/2,20/1,21/0,22/5,23/4,24/3,25/2,26/1}
    \draw (\x*.5,-1) rectangle node {\scriptsize \n} ++ (.5,.5);
  %equality
  \draw[<->] (10*.5+.1,-.25) -- node [above,yshift=-2.5] {\tiny $D\!-\!1$} ++ (5*.5-.2,0);
  \foreach \x in {15,...,20}
    \node at (\x*.5+.25,-.25) {$\downarrow$};
\end{tikzpicture}
%1069: 0   0   5   3   0   5   4   3   0   5   2   0   3   5   5   4   5   4   3   2   1   0   5   4   3   2   1
\caption{$D=6$, $\pi(1068)$ and $\pi(1069)$ are represented as sequences of height difference. We pictured a part of $s^{1069}$: a light grey square is a fired column and a dark grey square is a peak. By definition, $s^{1069}$ is dense starting from $\mathcal L(6,1069)$. We can notice that for every column $j$ such that $\mathcal L(6,1069)+D-1 \leq j < s^{1069},~ \pi(1068)_j=\pi(1069)_j$, which hints that the next avalanche reaching that part of the configuration will behave in the same way.}
\label{fig:6-1068-1069}
\end{figure}
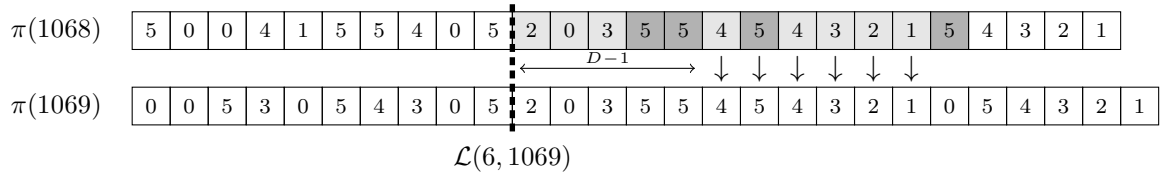

%\begin{wrapfigure}{l}{4.5cm}
\begin{figure}[!h]
  \begin{center}\input{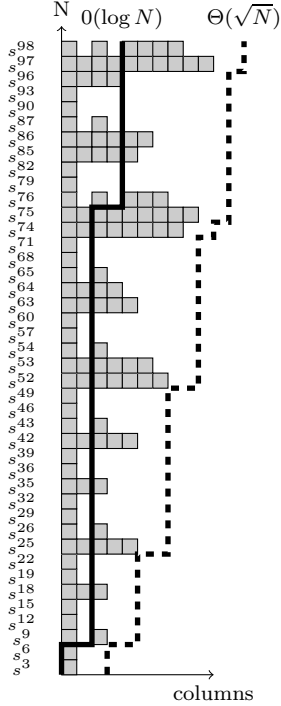}\end{center}
  \caption{$D=3$, one avalanche is depicted on each line. A grey square is a fired column. The black line illustrates the global density column $\mathcal L(3,N)$, and the dashed line depicts the maximal non-empty column. One can notice that the dashed line is 2 columns ahead of the rightmost fired column because it moves 1 grain to the new rightmost (maximal) non-empty column when applying the transition rule on that column).}
  \label{fig:gdc-3-0-100}
\end{figure}
%\end{wrapfigure}

%%%%%%%%%%%%%%%%%%%%%%%%%%%%%%%%%%
%
%   Application to KSPM(3)
%
%%%%%%%%%%%%%%%%%%%%%%%%%%%%%%%%%%

\subsection{Application to KSPM(3): $\mathcal L(3,N)$ in $O(\log N)$}\label{ss:avD=3}

In this section we prove that in KSPM(3), the global density column $\mathcal L(3,N)$ is in $O(\log N)$. Considering the $N$ first avalanches, from a logarithmic column in $N$ we can apply Theorem \ref{theorem:peak} and avalanches proceed pseudo locally, as described on figure \ref{fig:peak}.

\begin{proposition}\label{lemma:meta2}
  $[D=3]$\footnote{The proof given here is different from the proof presented in \cite{LATA}. We have more hope in the possible generalization of the present proof.}
  $\mathcal L(3,N)$ is in $O(\log N)$.
\end{proposition}

\begin{proof}
%%%%%%%%%%%%%%%%%%%%%%%%%%%%%%%%
%
%  DEBUT PREUVE
%
%%%%%%%%%%%%%%%%%%%%%%%%%%%%%%%%

We prove that $N$ is exponential in $\mathcal L(3,N)$. We use the fact that by definition of $\mathcal L(3,N)$ there exists an index $N' \leq N$ such that $\mathcal L(3,N)=\mathcal L'(3,N')$, which means that the $N'^{th}$ avalanche is dense starting from $\mathcal L(3,N)$ but not dense starting from $\mathcal L(3,N)-1$, and then prove that $N'$ is exponential in $\mathcal L'(3,N')$ (which proves the result since $N \geq N'$). For convenience and without loss of generality, let us consider that $N=N'$ (to avoid using the symbol $N'$ and saying that $N \geq N'$). Let $\ell=\mathcal L(3,N)$, the trick is that for the particular value $D=3$, the configuration $\pi(N-1)$ has unique and very regular values on columns $\pi(N-1)_0,\pi(N-1)_1,\dots,\pi(N-1)_\ell,\pi(N-1)_{\ell+1}$.

We suppose $\mathcal L'(3,N) > 0$, which is true for $N \geq 9$.

For convenience, let $\sigma=\pi(N-1)$ and $a=(a_0,a_1,\dots)$ be its {\em shot vector} {\em i.e.}, $a_i$ is the number of times that column $i$ has been fired\footnote{For example, if we study $\mathcal L(3,98)$ we can look at $s^{98}$ which is not dense starting from $\mathcal L(3,98)-1$ and dense starting from $\mathcal L(3,98)$, so $\mathcal L(3,98)=\mathcal L'(3,98)$ (visible on figure \ref{fig:gdc-3-0-100}). Then $\sigma=\pi(97)=(2,0,2,0,2,1,2,2,1,0,2,1,0^\omega)$ and $a=(41,14,21,12,11,7,5,3,2,1,0^\omega)$. Later in the proof we will have in this example that the sequence $x=(138,-68,34,-16,8,-3,2,0,1,0,0,1,0^\omega)$.}. According to the iteration rule we have the relation:
 $$\sigma_i=a_{i-2}-3a_i+2a_{i+1}$$
 {\em i.e.},   $$a_{i+1}=\frac{1}{2}(\sigma_i-a_{i-2}+3a_i)$$
%with the convention that $a_{-2}=N$ and $a_{-1}=0$.\\
We state  $A=\begin{pmatrix} 0 & 1 & 0\\ 0 & 0 & 1\\ -1/2 & 0 & 3/2 \end{pmatrix}$. %, and let $a$ denote the associated endomorphism of $\R^3$. 
We denote by $v_i$ the column vector such that  $v_i^T = ( 0,   0,  \sigma_i/2 ),$
and    $u_i$ the column vector such that $u_i^T =(a_{i-2},a_{i-1},a_i)$ (with the convention that  $u^T$ is  the line vector  obtained by transposition of the column vector $u$). The equality above can be algebraically written as
 $$u_{i+1}=A u_i + v_i$$
 with the initial condition 
 $$u_0 = (N-1, 0, a_0)$$

%The idea of the proof is to use the fact that, $u_{2j}$ is an integer vector enforces that $\vert \vert u_0 \vert \vert $ is exponential in $j$,  under our hypothesis.\\

\textit{Reduction of the dimension. } We first get a simplification, exploiting the fact that the value 1 is a double eigenvalue for $A$, as follows. Let $(e_1,e_2,e_3)$ be the canonical basis of $\Z^3$, we consider the vectors: 
 \begin{center}
 \begin{tabular}{l}
   $e'_1=e_1+e_2+e_3$\\
   $e'_2=e_2+2e_3$%\\
  % $e'_3=4e_1-2e_2+e_3$
 \end{tabular}
\end{center}

We have $A e'_1 = e'_1$ and $A e'_2 =   e'_1+ e'_2$. 
On the other hand, we have $e_1 = e'_1 - e'_2 + e_3$ and $e_2 =  e'_2 - 2 e_3$ (which guarantees that $(e'_1, e'_2, e_3)$ is a basis of $\Z^3$).

Let $P$ be the matrix of the projection on the vectorial 1-dimensional space $\Delta_{e_3}$, generated by $e_3$,  according to the direction of the plane $ \Pi_{e'_1, e'_2}$,  generated by  $e'_1$ and $e'_2$. We have: 
$$P =  \begin{pmatrix} 0 & 0 & 0 \\ 0 & 0 & 0 \\ 1 & -2 & 1  \end{pmatrix}$$
The plane $ \Pi_{e'_1, e'_2}$ is (globally) invariant by $A$, and, for  each integer $i$, $u_i - Pu_i$ is element of $ \Pi_{e'_1, e'_2}$. 

Thus $A(u_i - Pu_i)$ is element of $ \Pi_{e'_1, e'_2}$, which finally gives $PA(u_i - Pu_i) = 0$, {\em i.e.},  
$$PAu_i  = PA Pu_i. $$ 
So the equation $u_{i+1}=A u_i + v_i$ gives $Pu_{i+1}=PA u_i + Pv_i$, {\em i.e.} 
$$Pu_{i+1}=PA Pu_i + Pv_i$$
We state $x_i$ for the third component   of $Pu_i$ (both first and second components are null). 
The third component of $Pv_i$ is $\sigma_i/2$. We have 
$$P \, A  =  \begin{pmatrix} 0 & 0 & 0 \\ 0 & 0 & 0 \\- \frac {1}{2} & 1 & -\frac {1}{2}  \end{pmatrix}$$
Thus,  the equality $Pu_{i+1}=PA Pu_i + Pv_i$  yields to the equality:  
$$x_{i+1}=\frac { - x_i + \sigma_i}{2}$$
with initial condition $x_0 = N-1+a_0$.\\
 
\textit{Relative expression of elements $x_i$. } We can invert the previous equality to: 
\begin{eqnarray}\label{eq:xi}
x_{i}= - 2 x_{i+1} + \sigma_i
\end{eqnarray}

We claim that the above inequality can be generalized to get:
\begin{eqnarray}\label{eq:xi1k}
x_{i+1 -k }= (- 2)^k x_{i+1} + \sum_{r = 0}^{k-1} (-2)^{k-1-r} \sigma_{i -r}
\end{eqnarray}
This is easy by induction. We have the initialization for $k = 1$, and if we assume the result for $k$, then we have 
$$x_{i+1 -(k +1)} = - 2 x_{i+1-k} +\sigma_{i-k} =  - 2   ( (- 2)^k x_{i+1} + \sum_{r = 0}^{k-1} (-2)^{k-1-r} \sigma_{i -r}) + \sigma_{i-k}$$
$$x_{i -k }  = (- 2)^{k+1} x_{i+1} + \sum_{r = 0}^{k-1} (-2)^{k-r} \sigma_{i -r} + \sigma_{i-k} =  (- 2)^{k+1} x_{i+1} + \sum_{r = 0}^{k} (-2)^{k-r} \sigma_{i -r}$$
which is the result.\\

\textit{Specification in our case. } We claim that $\mathcal L'(3,N)$ is even {\em i.e.}, $\mathcal L'(3,N)=\mathcal L(3,N)=2j$ for some $j \in \mathbb N$ and that $\sigma$ needs to have as a prefix the sequence $2,0,2,0,2,0,2,\dots$ until column $2j$. That is, $2(02)^j$ is a prefix of $\sigma$. Indeed, the avalanche needs to propagate rightward until column $2j$ (otherwise $\mathcal L'(3,N) < 2j$) which is achieved thanks to values $2$ on even columns, and all the odd columns $1,3,\dots,2j-1$ must not be fired (otherwise the hypothesis of Theorem \ref{theorem:peak} is completed and the avalanche becomes dense) which is completed by values $0$ becoming $2$ when their right neighbor is fired. Any other value lets the avalanche become dense.

%Let $j$ be the largest integer such that $2(02)^j$ is a prefix of the configuration $\sigma$. We claim that $\mathcal L'(3,k) = \mathcal L(3,N) = 2j$. Indeed, $\mathcal L'(3,k) \geq 2j$ since applying the transition rules on the prefix of the form $2(02)^j$ fires every column of even index, and does not fire any column of odd index : the avalanche continues its rightward propagation thanks to values 2 on columns $0, 2, \dots, 2j$, and values 0 on columns $1, 3, \dots, 2j-1$ become 2 when columns $2, 4, \dots, 2j$ are fired. Moreover $\mathcal L'(3,k) \leq 2j$ since for any $\sigma_{2j+1}, \sigma_{2j+2}$ with $(\sigma_{2j+1},\sigma_{2j+2}) \neq (0,2)$, either
%\begin{itemize}
%  \item if column $2j+2$ is not fired then the avalanche stops;
%  \item if column $2j+2$ is fired then so is $2j+1$ since it implies that $\sigma_{2j+2}=2$ and hence $\sigma_{2j+1} > 0$, and column $2j+1$ receives 2 units of height difference from the firing of column $2j+2$.
%\end{itemize}
%In any of those cases, the avalanche is dense starting from $2j$ (immediately when the avalanche stops and from Theorem \ref{theorem:peak} in the other case), thus $\mathcal L'(3,k) \leq 2j$.

Applying the general formula (\ref{eq:xi1k}) with $i = 2j$ and $k = 2j+1$, we obtain: 
$$x_{2j+1 -(2j+1)  }= (- 2)^{2j+1} x_{2j+1} + \sum_{r = 0}^{2j} (-2)^{2j-r} \sigma_{2j-r}$$
 We have  $\sigma_{2j-r}= 0$ for $r $ odd and $\sigma_{2j-r}=2$ for $r$ even. Thus, stating $2s = 2j-r $,  we get 
\begin{eqnarray}\label{eq:x0}
x_0= (- 2)^{2j+1} x_{2j+1} + 2\sum_{s = 0}^{j} (-2)^{2s}  =  -2 (4)^j x_{2j+1} + 2 \frac{4^{j+1}-1}{3}
\end{eqnarray}

\textit{Conclusion. } We now prove that $x_0$ is exponentially large in $j$. Using the fact that $x_i$ is an integer for all $i$ ($x_i$ is the third component of $Pu_i$ with $u_i$ a vector of natural numbers and $P$ an integer matrix), we consider 3 cases:

\begin{enumerate}
  \item $\vert x_{2j+1} \vert \geq 2$. Then from equation (\ref{eq:xi}), for any $i \leq 2j+1$ the sign of $x_i$ is negative if and only if $i$ is odd (signs are alternating and $x_0$ is positive), therefore from (\ref{eq:x0}) we have $x_0 \geq 2 \frac{4^{j+1}-1}{3}$.
  \item $\vert x_{2j+1} \vert = 0$. Then equation (\ref{eq:x0}) becomes $x_0 = 2 \frac{4^{j+1}-1}{3}$.
  \item $\vert x_{2j+1} \vert = 1$. Since $\sigma_2j=2$, we have from equation (\ref{eq:xi}) that either $x_{2j}=4$ if $x_{2j+1}=-1$ or $x_{2j}=0$ if $x_{2j+1}=1$. In the former possibility we get the sign alternation and apply the reasoning of case 1, and in the latter possibility we also have $x_{2j-1}=0$ from equation (\ref{eq:xi}) because $\sigma_{2j-1}=0$ and apply the reasoning of case 2.
\end{enumerate}

In any case we get:
$$N-1+a_0 = x_0 \geq 2 \frac{4^j-1}{3}$$

We obviously have $a_0 \leq \frac{N-1}{2}$ since each firing of column $0$ moves 2 grains to the right. Thus, we get $\frac{3}{2}(N-1)  \geq  2 \frac{4^j-1}{3}$ from which we can easily deduce that:
$$j \leq \log_4{(\frac{9}{4}N-\frac{5}{4})}$$
hence,
$$\mathcal L(3, N)=2j \leq 2 \log_4{(\frac{9}{4}N-\frac{5}{4})}$$
which gives  the result.\qed
\end{proof}

For KSPM(3), after a short transient part of logarithmic length, the hypotheses of Theorem \ref{theorem:peak} are verified, and the study of avalanches can be turned into a pseudo linear process. Bounds on the maximal non-empty column $size(D,N)$ of a fixed point with $N$ grains shows that $size(D,N)$ is in $\Omega(\sqrt{N})$ when $D$ is constant. Indeed any height difference $\pi(N)_i$ verifies $0 \leq \pi(N)_i < D$ and it is easy to prove that there is no plateau\footnote{a {\em plateau} is a set of consecutive columns with the same height.} of length greater than $D$ columns, therefore $size(D,N)$ is lower bounded by $\sqrt{N/D}$ and upper bounded by $\sqrt{D N}$. As a consequence, the pseudo local process stands for the asymptotically complete behavior of avalanches.

Unfortunately, the approach above does not hold for $D >3$.  The main reason is that, for $D = 3$ unfired columns induce a very particular and {\em periodic} prefix ($2(02)^j$) on configurations. From $D = 4$, the structure of such a possible prefix is more complex and we did not yet get a tractable characterization of it. 

%%%%%%%%%%%%%%%%%%%%%%%%%%%%%%%%%%%%%%%%%%%%%%%%%%%%%%%%
%%%%%%%%%%%%%%%%%%%%%%%%%%%%%%%%%%%%%%%%%%%%%%%%%%%%%%%%
%%
%%   Transduction
%%
%%%%%%%%%%%%%%%%%%%%%%%%%%%%%%%%%%%%%%%%%%%%%%%%%%%%%%%%
%%%%%%%%%%%%%%%%%%%%%%%%%%%%%%%%%%%%%%%%%%%%%%%%%%%%%%%%

\section{Transduction}\label{s:transduction}

In this section we study the temporal regularities ---between successive avalanches--- that can be derived from the density of avalanches. Therefore, the developments below hold starting from the global density column $\mathcal L(D,N)$.

Firstly we study the similarities between two successive avalanches in subsection \ref{ss:successive}. We have seen in the previous section that when the $k^{th}$ avalanche has been performed, a large part of the fixed point $\pi(k)$ is equal to $\pi(k-1)$, hinting a similarity between $s^k$ and the next avalanche $s^{k+1}$. The similarity is stated in Lemma \ref{lemma:similar} by the mean of peaks: on a large part of the configuration, the peaks of $s^k$ and $s^{k+1}$ are exactly the same.

Secondly, we use the equality between successive sequences of peaks (a kind of temporal stability of the peaks) and combine it with our description of the avalanche process (Theorem \ref{theorem:peak}) to construct a finite state word transducer in subsection \ref{ss:transducer}. Given as an input a particular sequence of peaks: the greatest peaks between columns $i(D-1)$ and $(i+1)(D-1)$ for each of the $N$ first avalanches, an iteration of the transducer outputs the greatest peaks between columns $(i+1)(D-1)$ and $(i+2)(D-1)$ for each of the same $N$ first avalanches. A second iteration of the transducer outputs the greatest peaks between columns $(i+2)(D-1)$ and $(i+3)(D-1)$ for each of the $N$ first avalanches again. The {\em trace} up to $N$ on $I_i=\llbracket i(D-1);(i+1)(D-1)-1\rrbracket$ will be formally defined as the sequence of greatest peaks between columns $i(D-1)$ and $(i+1)(D-1)$ for each of the $N$ first avalanches (a subtlety to be explained later is that a peak is counted once during its ``lifetime'' {\em i.e.}, as long as it persists from an avalanche to the next one). Note that there may be more than one peak within columns of the interval $I_i$,  but at most one peak of maximal value (the value of a peak is of course its index), which we name ``greatest peak''. We can therefore write that given the trace up to $N$ on $I_i$, the transducer outputs the trace up to $N$ on $I_{i+1}$.

Note that the transducer does not perform complex computation, because given $N$ one can very easily compute the $N$ first avalanches by starting from the empty configuration, adding a grain and performing the first avalanche, adding a grain and performing the second avalanche, and so on. This computation is done temporally: we compute $s^1$, then $s^2$, then $s^3$, and so on. The interest of the transducer is that it embeds our knowledge on the regularity of avalanches and performs the computation spatially instead of temporally. Indeed, it computes the trace up to $N$ on $I_i$, then on $I_{i+1}$, then on $I_{i+2}$, and so on, each of the $I_j$ representing a spatial ``slice'' of avalanches, as depicted on figure \ref{fig:av}.  Now, the point is that if we know the trace up to $N$ on $I_j$, we are able to compute $\pi(N)_k$ for $k > (j+1)(D-1)$.

\begin{figure}
\begin{center}
\input{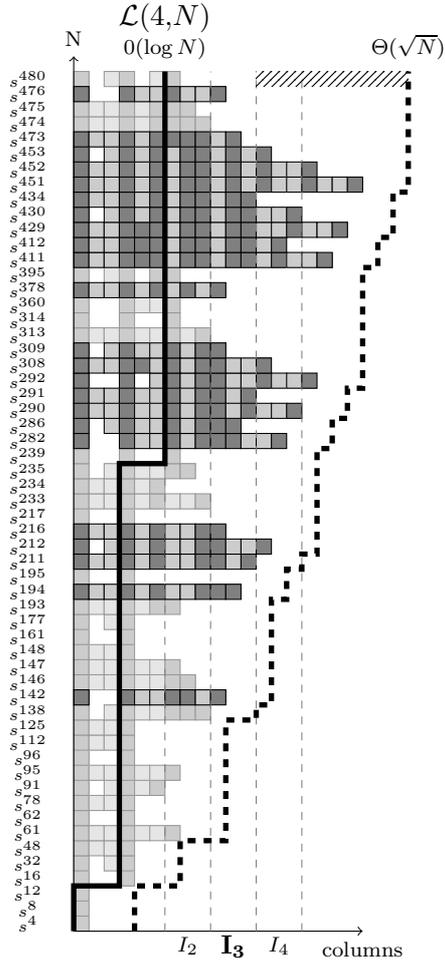}
\end{center}
\caption{D=4. Avalanches reaching $\mathcal L(D,N)$ up to 500, one per line. Since we exploit the density of every avalanche, an avalanche is meaningful for the study of the fixed point $\pi(500)$ if it goes beyond the global density column $\mathcal L(4,500)$ (and $D-1=3$ columns further for technical reason). Non-meaningful avalanches are shadowed. The black line is the global density column $\mathcal L(4,N)$, and the dashed line is the maximal non-empty column. A light grey square is a fired column, a dark grey square is a peak. The trace up to $500$ on $I_3$ is $0120120210$. From this trace, we can compute $\pi(500)_k$ for $k>12$ (12 is obtained as $(3+1)(D-1)$).\\
Each avalanche reaches a fixed point, so to every line $N$ representing the avalanche $s^N$ on the figure corresponds the fixed point $\pi(N)$, and to each position at ordinate $N$ and abscissa $k$ on the figure corresponds a difference of height $\pi(N)_k$. Following this convention, non empty columns of $\pi(500)_k$ for $k>12$ is hatched on the top right of the figure.\\
A precise study of the transducer for $D=3$ will lead to the prediction of a portion of fixed points (analogous to the one striped in this figure) which tends to 1 (subsection \ref{ss:D=3}).\\
{\scriptsize For technical reasons, the trace up to $N$ will be defined on slices on the left of column $\mathcal L(D,N)+2(D-1)$, therefore the trace up to 500 is formally not defined on $I_3$. We nevertheless use this example for its meaningfulness over size ratio.}}
\label{fig:av}
\end{figure}

The transducer depends on the parameter $D$ of the model, as explained in the construction at the end of Lemma \ref{lemma:transducer}. For $D=3$ and $N$ grains, we prove in subsection \ref{ss:waves} (Corollary \ref{corollary:decrease}) that for any trace up to $N$ on $I_i$, there exists an index $n$ in $O(\log N)$ such that the trace up to $N$ on $I_{i+n}$ is periodic. The proof uses only the transducer constructed for $D=3$. Hence, we can take for $i$ the smallest integer such that the transducer construction holds {\em i.e.}, $I_i$ is the first ``slice'' on the right of the global density column $\mathcal L(3,N)$ (plus 1 ``slice''  for technical reason). Therefore $i$ is in $O(\log N)$ as proved in Proposition \ref{lemma:meta2}, and the periodic trace up to $N$ emerges on $I_{i+n}$ with both $i$ and $n$ in $O(\log N)$. This study is concluded by Proposition \ref{proposition:wave}, which states that when we compute $\pi(N)_k$ for $k > (i+n+1)(D-1)$ from a regular trace, it presents a regular wavy shape. We can conclude for $D=3$ that the wave pattern emerges from a logarithmic column in the number of grain $N$, thus the fixed points are asymptotically completely wavy. 

%%%%%%%%%%%%%%%%%%%%%%%%%%%%%%%%%%
%
%   Successive avalanches
%
%%%%%%%%%%%%%%%%%%%%%%%%%%%%%%%%%%

\subsection{Successive avalanches}\label{ss:successive}

Toward the construction of the finite state word transducer, we start with a formal description of the intuitive implication of Theorem \ref{theorem:peak} that  two successive avalanches share a significant part of their peaks. Indeed, when the $k^{th}$ avalanche is dense from $l$ to $m$, for each column $i \in \llbracket l + D-1 ; m \llbracket$, columns $i-D+1$, $i$ and $i+1$ are fired within the $k^{th}$ avalanche. Therefore they both receive and give $D$ units of height difference, and $\pi(k)_i=\pi(k-1)^{\downarrow 0}_i=\pi(k-1)_i$ (figure \ref{fig:6-1068-1069}). Moreover, $\pi(k)_j = \pi(k-1)_j = 0$ for $j > m+D-1$. An intuitive consequence is that two consecutive avalanches are similar. This intuition is formally stated in this section.

Recall that $s^k$ denotes the $k^{th}$ avalanche of KSPM($D$). We define $\Phi(D,N)=(\phi^1,\dots,\phi^n)$, the subsequence of $(s^1,\dots,s^N)$ reaching column $\mathcal L(D,N)+D-1$. Formally, $s^k \in \Phi(D,N) \iff \mathcal L(D,N)+D-1 \in s^k$. $\Phi(D,N)$ is called the {\em sequence of long avalanches up to $N$} of KSPM($D$). An example is given on figure \ref{fig:long}.

\begin{figure}[!h]
  \begin{center}
    \begin{tikzpicture}[scale=.4]
  \node at (-1.3,0*.5+.25) {\tiny $s^{4}$};
  \node at (-1.3,1*.5+.25) {\tiny $s^{8}$};
  \node at (-1.3,2*.5+.25) {\tiny $s^{12}$};
  \node at (-1.3,3*.5+.25) {\tiny $s^{16}$};
  \node at (-1.3,4*.5+.25) {\tiny $s^{32}$};
  \node at (-1.3,5*.5+.25) {\tiny $s^{48}$};
  \node at (-1.3,6*.5+.25) {\tiny $s^{61}$};
  \node at (-1.3,7*.5+.25) {\tiny $s^{62}$};
  \node at (-1.3,8*.5+.25) {\tiny $s^{78}$};
  \node at (-1.3,9*.5+.25) {\tiny $s^{91}$};
  \node at (-1.3,10*.5+.25) {\tiny $s^{95}$};
  \node at (-1.3,11*.5+.25) {\tiny $s^{96}$};
  \node at (-1.3,12*.5+.25) {\tiny $s^{112}$};
  \node at (-1.3,13*.5+.25) {\tiny $s^{125}$};
  \node at (-1.3,14*.5+.25) {\tiny $s^{138}$};
  \node at (-1.3,15*.5+.25) {\tiny $s^{142}$};
  \node at (-1.3,16*.5+.25) {\tiny $s^{146}$};
  \node at (-1.3,17*.5+.25) {\tiny $s^{147}$};
  \node at (-1.3,18*.5+.25) {\tiny $s^{148}$};
  \node at (-1.3,19*.5+.25) {\tiny $s^{161}$};
  \node at (-1.3,20*.5+.25) {\tiny $s^{177}$};
  \node at (-1.3,21*.5+.25) {\tiny $s^{193}$};
  \node at (-1.3,22*.5+.25) {\tiny $s^{194}$};
  \node at (-1.3,23*.5+.25) {\tiny $s^{195}$};
  \filldraw[fill=black!50] (0*.5,0*.5) rectangle ++ (.5,.5);
  \filldraw[fill=black!50] (0*.5,1*.5) rectangle ++ (.5,.5);
  \filldraw[fill=black!50] (0*.5,2*.5) rectangle ++ (.5,.5);
  \filldraw[fill=black!50] (0*.5,3*.5) rectangle ++ (.5,.5);
  \filldraw[fill=black!50] (3*.5,3*.5) rectangle ++ (.5,.5);
  \filldraw[fill=black!50] (0*.5,4*.5) rectangle ++ (.5,.5);
  \filldraw[fill=black!20] (2*.5,4*.5) rectangle ++ (.5,.5);
  \filldraw[fill=black!50] (3*.5,4*.5) rectangle ++ (.5,.5);
  \filldraw[fill=black!50] (0*.5,5*.5) rectangle ++ (.5,.5);
  \filldraw[fill=black!20] (1*.5,5*.5) rectangle ++ (.5,.5);
  \filldraw[fill=black!20] (2*.5,5*.5) rectangle ++ (.5,.5);
  \filldraw[fill=black!50] (3*.5,5*.5) rectangle ++ (.5,.5);
  \filldraw[fill=black!50] (0*.5,6*.5) rectangle ++ (.5,.5);
  \filldraw[fill=black!20] (1*.5,6*.5) rectangle ++ (.5,.5);
  \filldraw[fill=black!20] (2*.5,6*.5) rectangle ++ (.5,.5);
  \filldraw[fill=black!50] (3*.5,6*.5) rectangle ++ (.5,.5);
  \filldraw[fill=black!20] (4*.5,6*.5) rectangle ++ (.5,.5);
  \filldraw[fill=black!20] (5*.5,6*.5) rectangle ++ (.5,.5);
  \filldraw[fill=black!50] (6*.5,6*.5) rectangle ++ (.5,.5);
  \filldraw[fill=black!50] (0*.5,7*.5) rectangle ++ (.5,.5);
  \filldraw[fill=black!50] (3*.5,7*.5) rectangle ++ (.5,.5);
  \filldraw[fill=black!50] (0*.5,8*.5) rectangle ++ (.5,.5);
  \filldraw[fill=black!20] (1*.5,8*.5) rectangle ++ (.5,.5);
  \filldraw[fill=black!20] (2*.5,8*.5) rectangle ++ (.5,.5);
  \filldraw[fill=black!50] (3*.5,8*.5) rectangle ++ (.5,.5);
  \filldraw[fill=black!50] (0*.5,9*.5) rectangle ++ (.5,.5);
  \filldraw[fill=black!20] (2*.5,9*.5) rectangle ++ (.5,.5);
  \filldraw[fill=black!50] (3*.5,9*.5) rectangle ++ (.5,.5);
  \filldraw[fill=black!20] (4*.5,9*.5) rectangle ++ (.5,.5);
  \filldraw[fill=black!50] (5*.5,9*.5) rectangle ++ (.5,.5);
  \filldraw[fill=black!50] (0*.5,10*.5) rectangle ++ (.5,.5);
  \filldraw[fill=black!20] (1*.5,10*.5) rectangle ++ (.5,.5);
  \filldraw[fill=black!20] (2*.5,10*.5) rectangle ++ (.5,.5);
  \filldraw[fill=black!50] (3*.5,10*.5) rectangle ++ (.5,.5);
  \filldraw[fill=black!20] (4*.5,10*.5) rectangle ++ (.5,.5);
  \filldraw[fill=black!20] (5*.5,10*.5) rectangle ++ (.5,.5);
  \filldraw[fill=black!50] (6*.5,10*.5) rectangle ++ (.5,.5);
  \filldraw[fill=black!50] (0*.5,11*.5) rectangle ++ (.5,.5);
  \filldraw[fill=black!50] (3*.5,11*.5) rectangle ++ (.5,.5);
  \filldraw[fill=black!50] (0*.5,12*.5) rectangle ++ (.5,.5);
  \filldraw[fill=black!20] (1*.5,12*.5) rectangle ++ (.5,.5);
  \filldraw[fill=black!20] (2*.5,12*.5) rectangle ++ (.5,.5);
  \filldraw[fill=black!50] (3*.5,12*.5) rectangle ++ (.5,.5);
  \filldraw[fill=black!50] (0*.5,13*.5) rectangle ++ (.5,.5);
  \filldraw[fill=black!20] (1*.5,13*.5) rectangle ++ (.5,.5);
  \filldraw[fill=black!20] (2*.5,13*.5) rectangle ++ (.5,.5);
  \filldraw[fill=black!50] (3*.5,13*.5) rectangle ++ (.5,.5);
  \filldraw[fill=black!50] (0*.5,14*.5) rectangle ++ (.5,.5);
  \filldraw[fill=black!20] (2*.5,14*.5) rectangle ++ (.5,.5);
  \filldraw[fill=black!50] (3*.5,14*.5) rectangle ++ (.5,.5);
  \filldraw[fill=black!20] (4*.5,14*.5) rectangle ++ (.5,.5);
  \filldraw[fill=black!20] (5*.5,14*.5) rectangle ++ (.5,.5);
  \filldraw[fill=black!50] (6*.5,14*.5) rectangle ++ (.5,.5);
  \filldraw[fill=black!50] (7*.5,14*.5) rectangle ++ (.5,.5);
  \filldraw[fill=black!50] (8*.5,14*.5) rectangle ++ (.5,.5);
  \filldraw[fill=black!50] (0*.5,15*.5) rectangle ++ (.5,.5);
  \filldraw[fill=black!50] (3*.5,15*.5) rectangle ++ (.5,.5);
  \filldraw[fill=black!20] (4*.5,15*.5) rectangle ++ (.5,.5);
  \filldraw[fill=black!20] (5*.5,15*.5) rectangle ++ (.5,.5);
  \filldraw[fill=black!50] (6*.5,15*.5) rectangle ++ (.5,.5);
  \filldraw[fill=black!50] (7*.5,15*.5) rectangle ++ (.5,.5);
  \filldraw[fill=black!20] (8*.5,15*.5) rectangle ++ (.5,.5);
  \filldraw[fill=black!50] (9*.5,15*.5) rectangle ++ (.5,.5);
  \filldraw[fill=black!50] (0*.5,16*.5) rectangle ++ (.5,.5);
  \filldraw[fill=black!20] (1*.5,16*.5) rectangle ++ (.5,.5);
  \filldraw[fill=black!20] (2*.5,16*.5) rectangle ++ (.5,.5);
  \filldraw[fill=black!50] (3*.5,16*.5) rectangle ++ (.5,.5);
  \filldraw[fill=black!20] (4*.5,16*.5) rectangle ++ (.5,.5);
  \filldraw[fill=black!20] (5*.5,16*.5) rectangle ++ (.5,.5);
  \filldraw[fill=black!50] (6*.5,16*.5) rectangle ++ (.5,.5);
  \filldraw[fill=black!50] (7*.5,16*.5) rectangle ++ (.5,.5);
  \filldraw[fill=black!50] (0*.5,17*.5) rectangle ++ (.5,.5);
  \filldraw[fill=black!20] (1*.5,17*.5) rectangle ++ (.5,.5);
  \filldraw[fill=black!20] (2*.5,17*.5) rectangle ++ (.5,.5);
  \filldraw[fill=black!50] (3*.5,17*.5) rectangle ++ (.5,.5);
  \filldraw[fill=black!20] (4*.5,17*.5) rectangle ++ (.5,.5);
  \filldraw[fill=black!20] (5*.5,17*.5) rectangle ++ (.5,.5);
  \filldraw[fill=black!50] (6*.5,17*.5) rectangle ++ (.5,.5);
  \filldraw[fill=black!50] (0*.5,18*.5) rectangle ++ (.5,.5);
  \filldraw[fill=black!20] (1*.5,18*.5) rectangle ++ (.5,.5);
  \filldraw[fill=black!20] (2*.5,18*.5) rectangle ++ (.5,.5);
  \filldraw[fill=black!50] (3*.5,18*.5) rectangle ++ (.5,.5);
  \filldraw[fill=black!50] (0*.5,19*.5) rectangle ++ (.5,.5);
  \filldraw[fill=black!50] (3*.5,19*.5) rectangle ++ (.5,.5);
  \filldraw[fill=black!50] (0*.5,20*.5) rectangle ++ (.5,.5);
  \filldraw[fill=black!20] (2*.5,20*.5) rectangle ++ (.5,.5);
  \filldraw[fill=black!50] (3*.5,20*.5) rectangle ++ (.5,.5);
  \filldraw[fill=black!50] (0*.5,21*.5) rectangle ++ (.5,.5);
  \filldraw[fill=black!20] (1*.5,21*.5) rectangle ++ (.5,.5);
  \filldraw[fill=black!20] (2*.5,21*.5) rectangle ++ (.5,.5);
  \filldraw[fill=black!50] (3*.5,21*.5) rectangle ++ (.5,.5);
  \filldraw[fill=black!20] (4*.5,21*.5) rectangle ++ (.5,.5);
  \filldraw[fill=black!50] (5*.5,21*.5) rectangle ++ (.5,.5);
  \filldraw[fill=black!50] (6*.5,21*.5) rectangle ++ (.5,.5);
  \filldraw[fill=black!50] (0*.5,22*.5) rectangle ++ (.5,.5);
  \filldraw[fill=black!20] (1*.5,22*.5) rectangle ++ (.5,.5);
  \filldraw[fill=black!20] (2*.5,22*.5) rectangle ++ (.5,.5);
  \filldraw[fill=black!50] (3*.5,22*.5) rectangle ++ (.5,.5);
  \filldraw[fill=black!20] (4*.5,22*.5) rectangle ++ (.5,.5);
  \filldraw[fill=black!50] (5*.5,22*.5) rectangle ++ (.5,.5);
  \filldraw[fill=black!20] (6*.5,22*.5) rectangle ++ (.5,.5);
  \filldraw[fill=black!20] (7*.5,22*.5) rectangle ++ (.5,.5);
  \filldraw[fill=black!50] (8*.5,22*.5) rectangle ++ (.5,.5);
  \filldraw[fill=black!50] (9*.5,22*.5) rectangle ++ (.5,.5);
  \filldraw[fill=black!50] (10*.5,22*.5) rectangle ++ (.5,.5);
  \filldraw[fill=black!50] (0*.5,23*.5) rectangle ++ (.5,.5);
  \filldraw[fill=black!50] (3*.5,23*.5) rectangle ++ (.5,.5);
  \draw[->] (0,0) -- ++ (4,0) node [right] {\scriptsize colums};
  \draw[->] (0,0) -- ++ (0,12.5) node [above] {\scriptsize $N$};
  \draw[line width=2pt] (3*.5,0) -- ++ (0,24*.5) node [above] {\scriptsize $\mathcal L(4,\!195)$};
  \draw[dashed] (6*.5,0) -- ++ (0,24*.5);
  \draw[<->] (3*.5,-.2) -- node [below] {\tiny $D\!-\!1$} ++ (3*.5,0);
\end{tikzpicture}
    \hspace{1cm}
    \begin{tikzpicture}[scale=.4]
  \node at (-1.3,6*.5+.25) {\tiny $\phi^1$};
  \node at (-1.3,10*.5+.25) {\tiny $\phi^2$};
  \node at (-1.3,14*.5+.25) {\tiny $\phi^3$};
  \node at (-1.3,15*.5+.25) {\tiny $\phi^4$};
  \node at (-1.3,16*.5+.25) {\tiny $\phi^5$};
  \node at (-1.3,17*.5+.25) {\tiny $\phi^6$};
  \node at (-1.3,21*.5+.25) {\tiny $\phi^7$};
  \node at (-1.3,22*.5+.25) {\tiny $\phi^8$};
  \filldraw[fill=black!50] (0*.5,6*.5) rectangle ++ (.5,.5);
  \filldraw[fill=black!20] (1*.5,6*.5) rectangle ++ (.5,.5);
  \filldraw[fill=black!20] (2*.5,6*.5) rectangle ++ (.5,.5);
  \filldraw[fill=black!50] (3*.5,6*.5) rectangle ++ (.5,.5);
  \filldraw[fill=black!20] (4*.5,6*.5) rectangle ++ (.5,.5);
  \filldraw[fill=black!20] (5*.5,6*.5) rectangle ++ (.5,.5);
  \filldraw[fill=black!50] (6*.5,6*.5) rectangle ++ (.5,.5);
  \filldraw[fill=black!50] (0*.5,10*.5) rectangle ++ (.5,.5);
  \filldraw[fill=black!20] (1*.5,10*.5) rectangle ++ (.5,.5);
  \filldraw[fill=black!20] (2*.5,10*.5) rectangle ++ (.5,.5);
  \filldraw[fill=black!50] (3*.5,10*.5) rectangle ++ (.5,.5);
  \filldraw[fill=black!20] (4*.5,10*.5) rectangle ++ (.5,.5);
  \filldraw[fill=black!20] (5*.5,10*.5) rectangle ++ (.5,.5);
  \filldraw[fill=black!50] (6*.5,10*.5) rectangle ++ (.5,.5);
  \filldraw[fill=black!50] (0*.5,14*.5) rectangle ++ (.5,.5);
  \filldraw[fill=black!20] (2*.5,14*.5) rectangle ++ (.5,.5);
  \filldraw[fill=black!50] (3*.5,14*.5) rectangle ++ (.5,.5);
  \filldraw[fill=black!20] (4*.5,14*.5) rectangle ++ (.5,.5);
  \filldraw[fill=black!20] (5*.5,14*.5) rectangle ++ (.5,.5);
  \filldraw[fill=black!50] (6*.5,14*.5) rectangle ++ (.5,.5);
  \filldraw[fill=black!50] (7*.5,14*.5) rectangle ++ (.5,.5);
  \filldraw[fill=black!50] (8*.5,14*.5) rectangle ++ (.5,.5);
  \filldraw[fill=black!50] (0*.5,15*.5) rectangle ++ (.5,.5);
  \filldraw[fill=black!50] (3*.5,15*.5) rectangle ++ (.5,.5);
  \filldraw[fill=black!20] (4*.5,15*.5) rectangle ++ (.5,.5);
  \filldraw[fill=black!20] (5*.5,15*.5) rectangle ++ (.5,.5);
  \filldraw[fill=black!50] (6*.5,15*.5) rectangle ++ (.5,.5);
  \filldraw[fill=black!50] (7*.5,15*.5) rectangle ++ (.5,.5);
  \filldraw[fill=black!20] (8*.5,15*.5) rectangle ++ (.5,.5);
  \filldraw[fill=black!50] (9*.5,15*.5) rectangle ++ (.5,.5);
  \filldraw[fill=black!50] (0*.5,16*.5) rectangle ++ (.5,.5);
  \filldraw[fill=black!20] (1*.5,16*.5) rectangle ++ (.5,.5);
  \filldraw[fill=black!20] (2*.5,16*.5) rectangle ++ (.5,.5);
  \filldraw[fill=black!50] (3*.5,16*.5) rectangle ++ (.5,.5);
  \filldraw[fill=black!20] (4*.5,16*.5) rectangle ++ (.5,.5);
  \filldraw[fill=black!20] (5*.5,16*.5) rectangle ++ (.5,.5);
  \filldraw[fill=black!50] (6*.5,16*.5) rectangle ++ (.5,.5);
  \filldraw[fill=black!50] (7*.5,16*.5) rectangle ++ (.5,.5);
  \filldraw[fill=black!50] (0*.5,17*.5) rectangle ++ (.5,.5);
  \filldraw[fill=black!20] (1*.5,17*.5) rectangle ++ (.5,.5);
  \filldraw[fill=black!20] (2*.5,17*.5) rectangle ++ (.5,.5);
  \filldraw[fill=black!50] (3*.5,17*.5) rectangle ++ (.5,.5);
  \filldraw[fill=black!20] (4*.5,17*.5) rectangle ++ (.5,.5);
  \filldraw[fill=black!20] (5*.5,17*.5) rectangle ++ (.5,.5);
  \filldraw[fill=black!50] (6*.5,17*.5) rectangle ++ (.5,.5);
  \filldraw[fill=black!50] (0*.5,21*.5) rectangle ++ (.5,.5);
  \filldraw[fill=black!20] (1*.5,21*.5) rectangle ++ (.5,.5);
  \filldraw[fill=black!20] (2*.5,21*.5) rectangle ++ (.5,.5);
  \filldraw[fill=black!50] (3*.5,21*.5) rectangle ++ (.5,.5);
  \filldraw[fill=black!20] (4*.5,21*.5) rectangle ++ (.5,.5);
  \filldraw[fill=black!50] (5*.5,21*.5) rectangle ++ (.5,.5);
  \filldraw[fill=black!50] (6*.5,21*.5) rectangle ++ (.5,.5);
  \filldraw[fill=black!50] (0*.5,22*.5) rectangle ++ (.5,.5);
  \filldraw[fill=black!20] (1*.5,22*.5) rectangle ++ (.5,.5);
  \filldraw[fill=black!20] (2*.5,22*.5) rectangle ++ (.5,.5);
  \filldraw[fill=black!50] (3*.5,22*.5) rectangle ++ (.5,.5);
  \filldraw[fill=black!20] (4*.5,22*.5) rectangle ++ (.5,.5);
  \filldraw[fill=black!50] (5*.5,22*.5) rectangle ++ (.5,.5);
  \filldraw[fill=black!20] (6*.5,22*.5) rectangle ++ (.5,.5);
  \filldraw[fill=black!20] (7*.5,22*.5) rectangle ++ (.5,.5);
  \filldraw[fill=black!50] (8*.5,22*.5) rectangle ++ (.5,.5);
  \filldraw[fill=black!50] (9*.5,22*.5) rectangle ++ (.5,.5);
  \filldraw[fill=black!50] (10*.5,22*.5) rectangle ++ (.5,.5);
  \draw[->] (0,0) -- ++ (4,0) node [right] {\scriptsize colums};
  \draw[->] (0,0) -- ++ (0,12.5) node [above] {\scriptsize $N$};
  \draw[line width=2pt] (3*.5,0) -- ++ (0,24*.5) node [above] {\scriptsize $\mathcal L(4,\!195)$};
  \draw[dashed] (6*.5,0) -- ++ (0,24*.5);
  \draw[<->] (3*.5,-.2) -- node [below] {\tiny $D\!-\!1$} ++ (3*.5,0);
\end{tikzpicture}
  \end{center}
  \caption{$D=4$. The 195 first avalanches are depicted on the left, only the long avalanches up to 195 are depicted on the right --- $s^k \in \Phi(4,195) \iff \mathcal L(4,195)+D-1 \in s^k$. The black line is the global density column $\mathcal L(4,N)$. A light grey square is a fired column, a dark grey square is a peak.}
  \label{fig:long}
\end{figure}
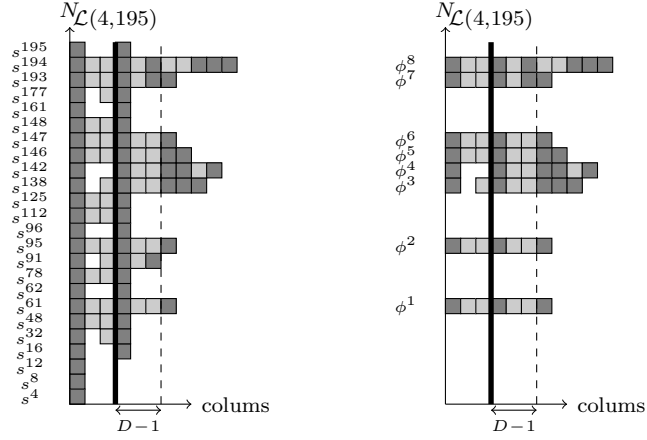

 We also define the corresponding sequence $(\mu^0, \mu^1, ....., \mu^n$) of fixed points such that $\mu^0 = \pi(0)  = 0^{\omega}$, and for each integer $k$, $\mu^k=\pi(m)$, where m satisfies $\phi^k=s^m$.
 
The definition of long avalanche is motivated by Theorem \ref{theorem:peak} which states that the avalanche process is understandable in simple terms on the right of the global density column, therefore we concentrate on firings on the right of the global density column and thus consider only avalanches reaching that part of the configuration.

\begin{remark}\label{remark:easy}
Remark that in KSPM($D$), if $s^k$ is a long avalanche up to $N$, whose sequence of peaks is denoted by $P^k$ (the largest peak being $\max P^k$), from Theorem \ref{theorem:peak} we have (figure \ref{fig:6-1068-1069}): 
\begin{itemize}
\item $\pi(k)_i = \pi(k-1)_i $ for $\mathcal L(D,N)+D-1 \leq i <\max P^k$;
\item $\pi(k)_{\max P^k} = \pi(k-1)_{\max P^k} -D+1 =    0$;
\item $\pi(k)_i =  \pi(k-1)_i +1 $ for $\max P^k <  i  \leq  \max P^k +D -1$;
\item $\pi(k)_i = \pi(k-1)_i $ for $ i > \max P^k +D -1$.  
\end{itemize}
\end{remark}

This comes from a clear application of transition rules (for each considered column $i$ we know which columns of the set $\{i-D+1, i , i+1\}$ are fired in $s^k$, so we can compute $\pi(k)_i$ from $\pi(k-1)$). In other words, The knowledge of $\max P^k$ allows to compute $\pi(k)$ from $\pi(k-1)$ straightforwardly.

\begin{lemma}\label{lemma:similar}
  In KSPM($D$), let $L$ be the global density column of $N$, and $\Phi=(\phi^1,\dots,\phi^n)$ its sequence of long avalanches up to $N$. Let $k < N$, and $P^k$ (resp. $P^{k+1}$) be the sequence of peaks $i$  of $\phi^k$ (resp. $\phi^{k+1}$) such that $i \geq L+2(D-1)$. We have: 
    $$P^k \setminus \{ \max P^k\}  ~~=~~ P^{k+1} \cap \llbracket L+2(D-1); \max P^k \llbracket.$$
\end{lemma}

The lemma above can be seen as follows: $\vert P^{k+1} \vert \geq \vert P^{k} \vert -1$,  and the  $ \vert P^{k} \vert -1$ first elements of  $ P^{k+1}$ and $  P^{k}$ are equal. Informally, the  peak sequence can increase in arbitrary manner, but can decrease only peak after peak.

\begin{proof}
  Let $\kappa$, $\kappa'$ be two integers such that $\phi^k$ is the $\kappa^{th}$ avalanche, and $\phi^{k+1}$ is the $\kappa'^{th}$ avalanche. For each column $i$ such that $i \in  \llbracket L+D-1;  \max P^k  \llbracket $, we have $i-D+1,i,i+1 \in \phi^k$ and therefore $\pi(\kappa)_i = \pi(\kappa-1)_i$.
  
  By definition of long avalanches, any avalanche $s$ between $\phi^k$ and $\phi^{k+1}$ stops before $L+D-1$, {\em i.e.} for all $i \geq L+D-1$, $i \notin s$. Combining it with previous remark, we have for all $\kappa'' \in \llbracket \kappa ; \kappa' \llbracket$:
  \begin{eqnarray}
    \text{for all } i \in \llbracket L+D-1; L+2(D-1) \llbracket & , & \pi(\kappa'')_i \geq \pi(\kappa-1)_i \label{lemma:similar:eq1}\\
    \text{for all } i \in \llbracket L+2(D-1); \max P^k \llbracket & , & \pi(\kappa'')_i = \pi(\kappa-1)_i \label{lemma:similar:eq2}
  \end{eqnarray}
  because columns within interval $\llbracket L+D-1; L+2(D-1) \llbracket$ can gain height difference when a column within $\llbracket L; L+D-1 \llbracket$ is fired. This is in particular true for $\kappa'' = \kappa'-1$.
  We now study the $\kappa'^{th}$ avalanche $\phi^{k+1}$. From relation \eqref{lemma:similar:eq1} and since $\pi(\kappa'-1)$ is a fixed point, for all $i \in \llbracket L+D-1; L+2(D-1) \llbracket,~ \pi(\kappa-1)_i = D-1 \Rightarrow \pi(\kappa'-1)_i = D-1$. Let $Q^k$ (resp. $Q^{k+1}$) be the sequence of peaks $i$ of $\phi^k$ (resp. $\phi^{k+1}$) such that  $i \in \llbracket L+D-1; L+2(D-1) \llbracket$,
  using Theorem \ref{theorem:peak} we therefore get
  \begin{eqnarray}
Q^{k}~~\subseteq~~ Q^{k+1}\label{lemma:similar:subset}
  \end{eqnarray}
  Let $I=\llbracket L+2(D-1); \max P^k \llbracket$. From relation \eqref{lemma:similar:eq2}
  \begin{eqnarray}
    \text{for all } i \in I ,~ \pi(\kappa-1)_i = D-1 \iff \pi(\kappa'-1)_i = D-1 \label{lemma:similar:eq3}
  \end{eqnarray}
  We now eventually prove the conclusion of the lemma. Let $p_{\underline i} = \min \{ i \in P^k \}$, from Theorem \ref{theorem:peak} we equivalently have $p_{\underline i} = \min \{ i' \in I | \pi(\kappa-1)_i=D-1 \}$ (the existence of $p_{\underline i}$ is a hypothesis of the lemma). Let $p'_{\underline i'} = \min \{  i' \in P^{k+1}¬¨‚Ä†\} = \min \{ i' \in I | \pi(\kappa'-1)_{i'} = D-1 \}$ (the existence of $p'_{\underline i'}$ is given by subset relation (\ref{lemma:similar:subset})), using relation (\ref{lemma:similar:eq3}) we have $p'_{\underline i'} = p_{\underline i}$.\\
  Other peaks within $I$ are obviously equal from Theorem \ref{theorem:peak} and relation (\ref{lemma:similar:eq2}) with $\kappa'' = \kappa'-1$.\qed
\end{proof}

%%%%%%%%%%%%%%%%%%%%%%%%%%%%%%%%%%
%
%   Transducer
%
%%%%%%%%%%%%%%%%%%%%%%%%%%%%%%%%%%

\subsection{Transducer}\label{ss:transducer}

We now exploit the similarity between successive avalanches. Intuitively, we will cut configurations into {\em intervals} $I_0,I_1,I_2,\dots$ of size $D-1$, describe what happens in each interval via the notion of {\em trace}, and relate the trace on $I_i$ to the trace on $I_{i+1}$ with a {\em finite state word transducer}.

The  {\em interval} $I_i$ is the column sequence $(i(D-1), i(D-1)+1,\dots,i(D-1)+D-2))$. We call {\em state} of an interval $I_i$ of a fixed point $\pi$ its value $(\pi_{i(D-1)}, \pi_{i(D-1)+1},\dots, \pi_{i(D-1)+D-2})$. Hence, each interval state is an element of the set  $\mathcal S = \{0,1,\dots, D-1\}^{D-1}$.  An obvious, nevertheless important remark is that in $\pi(0)$ the state of any interval is $(0,0,\dots,0)$. It will enable us to compute a fixed point from a trace. For convenience, we also write $00\dots0$ states and sequences when there is no ambiguity.
 
The formal definition of the {\em trace} up to $N$ on an interval $I_i$ is more involved. Wordily, it is the ordered sequence of relative positions of the largest peaks $p$ inside $I_i$ for the $N$ first long avalanches, where each peak is considered once while it stays from a long avalanche to the next one (this occurs when an avalanche performs firings strictly beyond $p$), see figure \ref{fig:av}. We formally define the trace in two steps: first the relative position of the largest peak inside $I_i$ ---the {\em type}--- then how we construct the trace by considering a peak only once.
 
We fix $i$ such that $i(D-1) \geq \mathcal L(D,N)+2(D-1)$ and consider the interval $I_i$. The maximal peak $j$ of $\phi^k$, such that $j < (i+1)(D-1)$, is denoted by $p(i, k)$. 
 The {\em type} $\alpha (i, k) $ of the long avalanche $\phi^k$ on $I_i$ is defined as: 
 \begin{itemize}
  \item $\alpha (i, k)  = p(i,k) \mod [D-1]$ if $p(i, k) \in  I_{i}$;
  \item $\alpha (i, k) = \epsilon$ if $p(i, k) \notin I_{i}$. 
\end{itemize}

Therefore, possible types are $\epsilon, 0, 1, \dots, D-2$. We formally consider types as words of length at most 1 over the alphabet $\mathcal T=\{0,1, \dots, D-2\}$, and $\epsilon$ as the empty word. So the set of types is $\mathcal T \cup \{\epsilon\}$. Nevertheless, it is more natural to think about types as letters.

Note that if  a long avalanche $\phi^k$ changes the state of $I_i$, then from Remark \ref{remark:easy} there necessarily exists a peak of $\phi^k$ in the interval $I_{i -1}$.

A {\em trace} is an element of $\mathcal T^*$. Intuitively, the trace up to $N$ on $I_i$ is a subsequence of the whole sequence of types $( \alpha(i,k) )_{k=1}^{k=N}$, where we keep only once the value in each consecutive sequence of equal values. For example the sequence $01121002112220$ should become $012102120$. Toward this definition we introduce some more notions given an interval $I_i$: type similarity between avalanches, subsequence of same type of avalanches, influent subsequence.

We say that two long avalanches $\phi^k$ and $\phi^{k'}$ are \emph{$i$-similar} if they have the same type on $I_i$ {\em i.e.}, if $\alpha(i,k)=\alpha(i,k')$. We can now divide the sequence $\Phi$ of long avalanches up to $N$ into maximal-length subsequences $(\phi^k, \phi^{k+1}, ...., \phi^{k''})$ such that, for each integer $k' \in \llbracket k ; k'' \llbracket$, $\phi^{k'}$ and $\phi^{k'+1}$  are $i$-similar. Such a subsequence is called an \emph{$i$-subsequence}. An $i$-subsequence is said of type $\alpha$ for $i$ when the type of each avalanche of the subsequence is $\alpha$. When  $\alpha$ is not the empty word $\epsilon$,  we say that the subsequence is \emph{$i$-influent}.

\begin{definition}
  The {\em trace} up to $N$ on $I_i$ is defined as the sequence of type of the $i$-influent-subsequences of $\Phi(D,N)$.
\end{definition}
Figure \ref{fig:trace} illustrates the definition of a trace.

\begin{figure}[!h]
  \begin{center}
    \begin{tikzpicture}[scale=.4]
  %part1
  \filldraw[fill=black!50] (0*.5,0*.5) rectangle ++ (.5,.5);
  \filldraw[fill=black!50] (0*.5,1*.5) rectangle ++ (.5,.5);
  \filldraw[fill=black!50] (1*.5,1*.5) rectangle ++ (.5,.5);
  \filldraw[fill=black!50] (0*.5,2*.5) rectangle ++ (.5,.5);
  \filldraw[fill=black!20] (1*.5,2*.5) rectangle ++ (.5,.5);
  \filldraw[fill=black!50] (2*.5,2*.5) rectangle ++ (.5,.5);
  \filldraw[fill=black!50] (0*.5,3*.5) rectangle ++ (.5,.5);
  \filldraw[fill=black!20] (1*.5,3*.5) rectangle ++ (.5,.5);
  \filldraw[fill=black!20] (2*.5,3*.5) rectangle ++ (.5,.5);
  \filldraw[fill=black!50] (3*.5,3*.5) rectangle ++ (.5,.5);
  \filldraw[fill=black!50] (0*.5,4*.5) rectangle ++ (.5,.5);
  \filldraw[fill=black!50] (1*.5,4*.5) rectangle ++ (.5,.5);
  \filldraw[fill=black!20] (2*.5,4*.5) rectangle ++ (.5,.5);
  \filldraw[fill=black!20] (3*.5,4*.5) rectangle ++ (.5,.5);
  \filldraw[fill=black!50] (4*.5,4*.5) rectangle ++ (.5,.5);
  \filldraw[fill=black!50] (0*.5,5*.5) rectangle ++ (.5,.5);
  \filldraw[fill=black!50] (1*.5,5*.5) rectangle ++ (.5,.5);
  \filldraw[fill=black!50] (0*.5,6*.5) rectangle ++ (.5,.5);
  \filldraw[fill=black!20] (1*.5,6*.5) rectangle ++ (.5,.5);
  \filldraw[fill=black!50] (2*.5,6*.5) rectangle ++ (.5,.5);
  \filldraw[fill=black!20] (3*.5,6*.5) rectangle ++ (.5,.5);
  \filldraw[fill=black!20] (4*.5,6*.5) rectangle ++ (.5,.5);
  \filldraw[fill=black!50] (5*.5,6*.5) rectangle ++ (.5,.5);
  \filldraw[fill=black!50] (0*.5,7*.5) rectangle ++ (.5,.5);
  \filldraw[fill=black!20] (1*.5,7*.5) rectangle ++ (.5,.5);
  \filldraw[fill=black!50] (2*.5,7*.5) rectangle ++ (.5,.5);
  \filldraw[fill=black!50] (0*.5,8*.5) rectangle ++ (.5,.5);
  \filldraw[fill=black!20] (1*.5,8*.5) rectangle ++ (.5,.5);
  \filldraw[fill=black!20] (2*.5,8*.5) rectangle ++ (.5,.5);
  \filldraw[fill=black!50] (3*.5,8*.5) rectangle ++ (.5,.5);
  \filldraw[fill=black!20] (4*.5,8*.5) rectangle ++ (.5,.5);
  \filldraw[fill=black!20] (5*.5,8*.5) rectangle ++ (.5,.5);
  \filldraw[fill=black!50] (6*.5,8*.5) rectangle ++ (.5,.5);
  \filldraw[fill=black!50] (0*.5,9*.5) rectangle ++ (.5,.5);
  \filldraw[fill=black!20] (1*.5,9*.5) rectangle ++ (.5,.5);
  \filldraw[fill=black!20] (2*.5,9*.5) rectangle ++ (.5,.5);
  \filldraw[fill=black!50] (3*.5,9*.5) rectangle ++ (.5,.5);
  \filldraw[fill=black!50] (0*.5,10*.5) rectangle ++ (.5,.5);
  \filldraw[fill=black!50] (0*.5,11*.5) rectangle ++ (.5,.5);
  \filldraw[fill=black!20] (0*.5,12*.5) rectangle ++ (.5,.5);
  \filldraw[fill=black!50] (1*.5,12*.5) rectangle ++ (.5,.5);
  \filldraw[fill=black!50] (2*.5,12*.5) rectangle ++ (.5,.5);
  \filldraw[fill=black!20] (3*.5,12*.5) rectangle ++ (.5,.5);
  \filldraw[fill=black!50] (4*.5,12*.5) rectangle ++ (.5,.5);
  \filldraw[fill=black!20] (5*.5,12*.5) rectangle ++ (.5,.5);
  \filldraw[fill=black!20] (6*.5,12*.5) rectangle ++ (.5,.5);
  \filldraw[fill=black!50] (7*.5,12*.5) rectangle ++ (.5,.5);
  \filldraw[fill=black!20] (0*.5,13*.5) rectangle ++ (.5,.5);
  \filldraw[fill=black!50] (1*.5,13*.5) rectangle ++ (.5,.5);
  \filldraw[fill=black!50] (2*.5,13*.5) rectangle ++ (.5,.5);
  \filldraw[fill=black!20] (3*.5,13*.5) rectangle ++ (.5,.5);
  \filldraw[fill=black!50] (4*.5,13*.5) rectangle ++ (.5,.5);
  \filldraw[fill=black!20] (0*.5,14*.5) rectangle ++ (.5,.5);
  \filldraw[fill=black!50] (1*.5,14*.5) rectangle ++ (.5,.5);
  \filldraw[fill=black!50] (2*.5,14*.5) rectangle ++ (.5,.5);
  \filldraw[fill=black!20] (3*.5,14*.5) rectangle ++ (.5,.5);
  \filldraw[fill=black!20] (4*.5,14*.5) rectangle ++ (.5,.5);
  \filldraw[fill=black!50] (5*.5,14*.5) rectangle ++ (.5,.5);
  \filldraw[fill=black!20] (6*.5,14*.5) rectangle ++ (.5,.5);
  \filldraw[fill=black!20] (7*.5,14*.5) rectangle ++ (.5,.5);
  \filldraw[fill=black!50] (8*.5,14*.5) rectangle ++ (.5,.5);
  \filldraw[fill=black!20] (0*.5,15*.5) rectangle ++ (.5,.5);
  \filldraw[fill=black!50] (1*.5,15*.5) rectangle ++ (.5,.5);
  \filldraw[fill=black!50] (2*.5,15*.5) rectangle ++ (.5,.5);
  \filldraw[fill=black!20] (3*.5,15*.5) rectangle ++ (.5,.5);
  \filldraw[fill=black!20] (4*.5,15*.5) rectangle ++ (.5,.5);
  \filldraw[fill=black!50] (5*.5,15*.5) rectangle ++ (.5,.5);
  \filldraw[fill=black!20] (0*.5,16*.5) rectangle ++ (.5,.5);
  \filldraw[fill=black!50] (1*.5,16*.5) rectangle ++ (.5,.5);
  \filldraw[fill=black!50] (2*.5,16*.5) rectangle ++ (.5,.5);
  \filldraw[fill=black!20] (0*.5,17*.5) rectangle ++ (.5,.5);
  \filldraw[fill=black!50] (1*.5,17*.5) rectangle ++ (.5,.5);
  \filldraw[fill=black!20] (2*.5,17*.5) rectangle ++ (.5,.5);
  \filldraw[fill=black!50] (3*.5,17*.5) rectangle ++ (.5,.5);
  \filldraw[fill=black!20] (4*.5,17*.5) rectangle ++ (.5,.5);
  \filldraw[fill=black!20] (5*.5,17*.5) rectangle ++ (.5,.5);
  \filldraw[fill=black!50] (6*.5,17*.5) rectangle ++ (.5,.5);
  \filldraw[fill=black!20] (7*.5,17*.5) rectangle ++ (.5,.5);
  \filldraw[fill=black!20] (8*.5,17*.5) rectangle ++ (.5,.5);
  \filldraw[fill=black!50] (9*.5,17*.5) rectangle ++ (.5,.5);
  \filldraw[fill=black!20] (0*.5,18*.5) rectangle ++ (.5,.5);
  \filldraw[fill=black!50] (1*.5,18*.5) rectangle ++ (.5,.5);
  \filldraw[fill=black!20] (2*.5,18*.5) rectangle ++ (.5,.5);
  \filldraw[fill=black!50] (3*.5,18*.5) rectangle ++ (.5,.5);
  \filldraw[fill=black!20] (4*.5,18*.5) rectangle ++ (.5,.5);
  \filldraw[fill=black!20] (5*.5,18*.5) rectangle ++ (.5,.5);
  \filldraw[fill=black!50] (6*.5,18*.5) rectangle ++ (.5,.5);
  \filldraw[fill=black!20] (0*.5,19*.5) rectangle ++ (.5,.5);
  \filldraw[fill=black!50] (1*.5,19*.5) rectangle ++ (.5,.5);
  \filldraw[fill=black!20] (2*.5,19*.5) rectangle ++ (.5,.5);
  \filldraw[fill=black!50] (3*.5,19*.5) rectangle ++ (.5,.5);
  \filldraw[fill=black!20] (0*.5,20*.5) rectangle ++ (.5,.5);
  \filldraw[fill=black!50] (1*.5,20*.5) rectangle ++ (.5,.5);
  \filldraw[fill=black!50] (0*.5,21*.5) rectangle ++ (.5,.5);
  \draw[dashed,black!50] (0,0) -- ++ (0,24*.5);
  \draw[dashed,black!50] (3*.5,0) -- ++ (0,24*.5);
  \draw[dashed,black!50] (6*.5,0) -- ++ (0,24*.5);
  \draw[dashed,black!50] (9*.5,0) -- ++ (0,24*.5);
  \node at (1.5*.5,-.8) {\scriptsize $I_3$};
  \node at (4.5*.5,-.8) {\small $\mathbf{I_4}$};
  \node at (7.5*.5,-.8) {\scriptsize $I_5$};
  \node at (4.5*.5,-2) {\textcircled{1}};
  %part2
  \def\o2{8}
  \foreach \y/\t in
{0/\epsilon,1/\epsilon,2/\epsilon,3/0,4/1,5/\epsilon,6/2,7/\epsilon,8/0,9/0,10/\epsilon,11/\epsilon,12/1,13/1,14/2,15/2,16/\epsilon,17/0,18/0,19/0,20/\epsilon,21/\epsilon}
    \node at (\o2,\y*.5+.25) {\scriptsize $\t$};
  \node at (\o2,-2) {\textcircled{2}};
  \foreach \a/\b in {0/3,8/10,10/12,12/14,14/16,17/20,20/22}
    \draw[decorate, decoration=brace] (\o2+.5,\b*.5-.1) -- (\o2+.5,\a*.5+.1);
  %part3
  \def\o3{10}
  \foreach \y/\t in
{1/\epsilon,3/0,4/1,5/\epsilon,6/2,7/\epsilon,8.5/0,10.5/\epsilon,12.5/1,14.5/2,16/\epsilon,18/0,20.5/\epsilon}
    \node at (\o3,\y*.5+.25) {\scriptsize $\t$};
  \node at (\o3,-2) {\textcircled{3}};
  %part4
  \def\o4{12}
  \foreach \y/\t in
{3/0,4/1,6/2,8.5/0,12.5/1,14.5/2,18/0}
    \node at (\o4,\y*.5+.25) {\scriptsize $\t$};
  \node at (\o4,-2) {\textcircled{4}};
\end{tikzpicture}
  \end{center}
  \caption{$D=4$. Illustration of definitions around the trace up to 500 on $I_4$. The left datas are extracted from figure \ref{fig:av}. \textcircled{1} Long avalanches up to $500$, a dark grey square is a peak, a light grey square is fired column. $\mathcal L(4,500)=6$ and $\Phi(4,500)=(\phi^1,\dots,\phi^{22})$. \textcircled{2} Corresponding whole sequence of types on $I_4$: $( \alpha(4,k) )_{k=1}^{k=22}$. \textcircled{3} $4$-similar long avalanches are grouped into $4$-subsequences. \textcircled{4} The trace up to 500 on $I_4$ is obtain with the $4$-influent-subsequences: 0120120.}
  \label{fig:trace}
\end{figure}

Note that from Lemma \ref{lemma:similar}, each $(i+1)$-influent-subsequence is contained in an $i$-influent-subsequence. Indeed, the peak sequence can be increased in an arbitrary manner but only the greatest peak can disappear from a long avalanche to the next one. Therefore if we consider an $i$-subsequence of type $\alpha$, $\Phi_{[k,k'']}=(\phi^k,\dots,\phi^{k''})$, $\Phi_{[k,k'']}$ ends only if the greatest peak of $\phi^{k''}$ is at relative position $\alpha$ in $I_i$. Also, for the $i$-subsequence $\Phi_{[k,k'']}$ to start, a peak must appear on relative position $\alpha$ in $I_i$ during $\phi^k$, and it enforces the previous avalanche $\phi^{k-1}$ to stop before this position ($\max \phi^{k-1} < i(D-1) + \alpha$) because from Remark \ref{remark:easy} the only changes induced by $\phi^{k-1}$ on height differences which can create a peak are on the $D-1$ columns following its largest peak. As a consequence $\phi^{k-1}$ is of type $\epsilon$ on $I_{i+1}$ and so does not belong to an $(i+1)$-influent-subsequence. We can conclude that any $(i+1)$-influent-subsequence must start and end within an $i$-influent subsequence.

The finite state word transducer is constructed to compute the sequence of $(i+1)$-influent-subsequences within an $i$-influent-subsequence, hence relating the trace on $I_i$ to the trace on $I_{i+1}$.

\begin{lemma}\label{lemma:transducer}
Let  $\Phi_{[k,k'']} = (\phi^k,\phi^{k+1}, ...., \phi^{k''})$ be an $i$-subsequence of type  $\alpha$, with $k'' \leq N$, and with $I_{i+1}$ an interval whose columns are greater than $\mathcal L(D,N)+3(D-1)$. 
Given the state $(a_0, a_1, ..., a_{D-2})$ of $I_{i+1}$ in the configuration $\mu^{k-1}$, and  $\alpha$,  one can compute, with no need of more knowledge: 
\begin{itemize}
\item the state $(a'_0, a'_1, ..., a'_{D-2})$ of $I_{i+1}$  in the configuration $\mu^{k''}$, 
\item the sequence of  types  of  the successive $(i+1)$-influent subsequences contained in $\Phi_{[k,k'']}$.
\end{itemize}
\end{lemma}

\begin{proof}
This is obvious when the type of the subsequence is $\epsilon$, since there is no change and the $(i+1)$-subsequence contained in $(\phi^k,\dots,\phi^{k''})$ is also $\epsilon$.

 The computation is simple when there is no integer $m \in \llbracket 0 ; \alpha \rrbracket$ such that $a_m = D-1$. In this case,  the peak $p(i, k)$ is the last peak of  $\phi^k$,  thus  $\mu^{k}_{p(i, k)} = 0$, which gives that  $p(i, k)$ is not a peak of $\phi^{k+1}$, thus the subsequence is reduced to a singleton which is not $(i+1)$-influent (second part of the result). For $(D-1)(i+1) \leq j \leq p(i, k)+ D-1$, we have $\mu^{k}_{j} = \mu^{k-1}_{j} +1$, and for $p(i, k)+ D-1 < j <(D-1)(i+2)$, we have $\mu^{k}_{j} = \mu^{k-1}_{j}$. Thus, we have $a'_m = a_m +1 $ for $m \in \llbracket 0 ; \alpha \rrbracket$ and  $a'_m = a_m $ for $m \in \rrbracket \alpha ; D-2 \rrbracket$ (first part of the result). 
 
Otherwise, $\phi^k$ contains a peak in  $I_{i+1}$.  Let $q(i+1, k)$ denote the largest one. The column $q(i+1, k)$ is the largest $j$  such that $\mu^{k-1}_j = D-1$ and $j <(D-1)(i+2)$. Thus $q(i+1, k) \mod{D-1}$ is the largest $m$  such that $a_m = D-1$. 
 In this case,  $\phi^k$ starts an $(i+1)$-subsequence of type  $q(i+1, k)$. Consider the following long avalanches. From Lemma \ref{lemma:similar}, while  
$q(i+1, k)$ remains a peak of $\phi^{k'}$, $p(i, k)$ also remains a peak of $\phi^{k'}$. 
From Remark \ref{remark:easy}, while $q(i+1, k)$ is not the last peak of $\phi^{k'}$, the state of $I_{i+1}$
remains invariant. 
So the first avalanche $\phi^{k'}$ that changes the state of $I_{i+1}$ is the one whose last peak is $q(i+1, k)$ ($k' \leq k''$ since $\Phi_{[k,k'']}$ is a complete $i$-subsequence). We have  $\mu^{k'}_{q(i+1, k)}  = 0$, 
which closes the $(i+1)$-subsequence of type $q(i+1, k)$. 
We also have  $\mu^{k'}_j = \mu^{k}_j +1$ for  $q(i+1, k) < j < (D-1)(i+2)$, and $\mu^{k'}_j = \mu^{k}_j$ for $p(i, k) \leq j < q(i+1, k)$. This gives the state of $I_{i+1}$ for $\mu^{k'}$  (as in the previous case, this can be rewritten to show that this state can be expressed only from $\alpha$ and  $(a_0, a_1, ..., a_{D-2})$) and proves that $p(i, k) = p(i, k'+1)$. 

The argument above can be repeated as long as we have a column $j$ of $I_i$ whose current value is $D -1$. When there is no more such column, the peak $p (i, k)$ is deleted (its value becomes 0) by the next long avalanche which is necessarily $\phi^{k''}$ from the maximality of $i$-similar subsequences.\qed
\end{proof}

The algorithm below gives the exact computation underlying the proof Lemma \ref{lemma:transducer}. From the state of an interval $I_{i+1}$ and an avalanche type on $I_i$, $f$ returns the greatest fired peak in $I_{i+1}$, and $g$ computes the new state of $I_{i+1}$ and appends the result of $f$ to a sequence of types on interval $I_{i+1}$. $g$ recursively calls itself, anticipating the $i$-similarity of successive avalanches when $\max P^k$ lies on the right of interval $I_i$.

%%%%%%%%%%%%%%%%%%%%%%%%%%%%%%%%%%
%   Algorithm
%%%%%%%%%%%%%%%%%%%%%%%%%%%%%%%%%%

\vspace{.2cm}
\noindent
$\left[ ~\parbox{\textwidth}{
\textbf{Input:} a non empty type $\alpha$ and an interval state $A=(a_0,\dots,a_{D-2})$.\\
\textbf{Data structure:} a sequence $u$ of types.\\
\textbf{Functions:}\\
$\begin{array}{>{\raggedright}p{.33\textwidth} | >{\raggedright}p{.6\textwidth}}
$f : \mathcal S \times \mathcal T \to \mathcal T \cup \{\epsilon\}$ & $g : \mathcal S \times \mathcal T \times \mathcal T^* \to \mathcal S \times \mathcal T^*$\tabularnewline
\hline
$f(A,\alpha) :=$\\
\textbf{if} $( \{ m \leq \alpha | a_m=D-1¬†\} \!\neq\! \emptyset )$\\
\textbf{then}\\
\hspace{.2cm} $\max \{¬†m | a_m = D-1\}$\\
\textbf{else}\\
\hspace{.2cm} $\epsilon$
&
$g(A,\alpha,u) :=$\\
\textbf{match} $f(A,\alpha)$ \textbf{with}\\
\hspace{.2cm} $| \epsilon \to (a_0+1,\dots,a_{\alpha}+1,a_{\alpha+1},\dots,a_{D-2}),u)$\\
\hspace{.2cm} $| p \to g((a_0,\dots,a_{p-1},0,a_{p+1}+1,\dots,a_{D-2}+1),\alpha,u\!::\!p)$
\end{array}$\\
\textbf{Computation:} $(A,\alpha) \mapsto g(A,\alpha,\epsilon)$
}\right.$
\vspace{.2cm}

The algorithm above allows to define a deterministic finite state transducer \textswab T (see for example \cite{berstel}) computes the trace up to $N$ on $I_{i+1}$ given the trace up to $N$ on $I_i$. It is  

\begin{definition}
The finite state word transducer \textswab T for $D$ is a a 5-tuple $(Q,\Sigma,\Gamma,I,\delta)$ where:
\begin{itemize}
  \item the set of states $Q$ is $\mathcal S$;
  \item the input and output alphabets (resp. $\Sigma$ and $\Gamma$) are equal to $\mathcal A= \mathcal T \setminus \{ \epsilon \} = \{ 0, \dots, D-2 \}$;
  \item the transition function $\delta$ has type $Q \times \Sigma \to Q \times \Gamma^*$ and is defined by the algorithm above: $\delta(A,\alpha)=Computation(A,\alpha)$;
  \item the initial state is $(0, 0,\dots, 0)$, and we do not need to define a final state.
\end{itemize}
\end{definition}

The image of a word $u$ by \textswab T is denoted by $t(u)$. The formal definition of $t$ will not be useful and is a bit involved, we nevertheless give it for the sake of formalness:
\begin{definition}
  In order to define $t$, we need to define some other functions. Let $\delta_{\mathcal S}$ (resp. $\delta_{\mathcal T}$) denote the first (resp. second) projection of the result of the application of $\delta$:
  $$\text{if } \delta(A,\alpha)=(A',u') \text{, then }
  \begin{array}[t]\{{l}.
    \delta_{\mathcal S}(A,\alpha)=A'\\
    \delta_{\mathcal T}(A,\alpha)=u'
  \end{array}$$
  We also define the generalized transition function of the transducer:
  $$\delta^*(A,u)=\delta^\circ(A,u,\epsilon) \text{ with }
  \begin{array}[t]\{{l}.
    \delta^\circ(A,\alpha u,u')=\delta^\circ(\delta_{\mathcal S}(A,\alpha),u,u'\delta_{\mathcal T}(A,\alpha))\\
    \delta^\circ(A,\epsilon,u')=(A,u')
  \end{array}$$
  Then $t$ is defined for all word $u \in \mathcal A^*$ by
  $$t(u)=\delta^*_{\mathcal T}(00\dots0,u)$$
  where $\delta^*_{\mathcal T}$ is the second projection of the result of the application of $\delta^*$.
\end{definition}

As presented in Lemma \ref{lemma:transducer}, if $\alpha$ is the type of an $i$-influent-subsequence $\Phi_{[k,k']}$, and $A$ the state of $I_{i+1}$ in $\mu^{k-1}$, then $\delta(A,\alpha)$ computes the state of $I_{i+1}$ in $\mu^{k'}$ and the sequence of types of the $(i+1)$-influent-subsequences within $\Phi_{[k,k']}$. It follows that if $\beta$ is the type of the $i$-influent-subsequence $\Phi_{[k'+1,k'']}$, we already know the state of $I_{i+1}$ in $\mu^{k'}$ and are able to compute the sequence of types of the $(i+1)$-influent-subsequences within $\Phi_{[k'+1,k'']}$.

We have already highlighted the fact that, in the fixed point $\pi(0)$, any interval $I_i$ is in the state $(0,0,\dots,0)$. As a consequence, if $u$ is the trace up to $N$ on $I_i$, then $t(u)$ is the trace up to $N$ on $I_{i+1}$. Furthermore, since the input and output alphabets of \textswab T are equal, we have that $t(t(u))$ is the trace up to $N$ on $I_{i+2}$ and more generally that $t^n(u)$ is the trace up to $N$ on $I_{i+n}$. This has been illustrated on figure \ref{fig:intro}, where $u$ is the trace up to $N$ on $\mathcal L(D,N)+3(D-1)$.

Note that, more generally, the last subsequence of a trace may not be completed. This detail will be discussed in subsection \ref{ss:waves}.

\begin{figure}[!h]
  \begin{center}
    \begin{tikzpicture}[scale=1.5]
  \node[circle,draw,line width=1pt,fill=black!30] (21) at (0,0) {21};
  \node[circle,draw,line width=1pt,fill=black!30] (11) at (4,0) {11}
    edge [->,line width=1pt,out=-170,in=-10] node [below=-3pt] {$a|\epsilon$} (21)
    edge [<-,line width=1pt,out=170, in= 10] node [above=-3pt] {$b|ab$} (21);
  \node[circle,draw,line width=1pt,fill=black!30] (12) at (0,-4) {12}
    edge [->,line width=1pt,out=100,in=-100] node [left=-3pt] {$b|b$} (21)
    edge [<-,line width=1pt,out=80,in=-80] node [right=-3pt] {$a|a$} (21);
  \node[circle,draw,line width=1pt,fill=black!30] (22) at (4,-4) {22}
    edge [->,line width=1pt,out=-170,in=-10] node [below=-3pt] {$b|ba$} (12)
    edge [<-,line width=1pt,out=170,in=10] node [above=-3pt] {$a|\epsilon$} (12)
    edge [->,line width=1pt,out=100,in=-100] node [left=-3pt] {$a|ba$} (11)
    edge [<-,line width=1pt,out=80,in=-80] node [right=-3pt] {$b|\epsilon$} (11);
  \node[circle,draw,line width=1pt,fill=black!10] (20) at (2,-2) {20}
    edge [->,line width=1pt] node[below=-3pt] {$b|\epsilon$} (12)
    edge [->,line width=1pt] node[below=-3pt] {$a|\epsilon$} (11);
  \node[circle,draw,line width=1pt,fill=black!10] (10) at (1,-1) {10}
    edge [->,line width=1pt] node[below=-3pt] {$b|\epsilon$} (21)
    edge [->,line width=1pt] node[below=-3pt] {$a|\epsilon$} (20);
  \node[circle,draw,double,line width=1pt,fill=black!10] (00) at (2.5,-.5) {00}
    edge [->,line width=1pt] node[below=-3pt] {$b|\epsilon$} (11)
    edge [->,line width=1pt] node[below=-3pt] {$a|\epsilon$} (10);
\end{tikzpicture}
  \end{center}
  \caption{Transducer for $D=3$ - Edges are labelled $x|u$, where $x \in \mathcal A$ is the type to the current interval (input) and $u \in \mathcal A^*$ is the resulting sequence of types applied to the next interval (output). 00 is the initial state. For example, $t(abaaaaab) = abaab$. Remark  that,  for $n >0$, we have:  $t((ab)^n) = (ab)^{n-1}$. }
  \label{fig:transducer}
\end{figure}
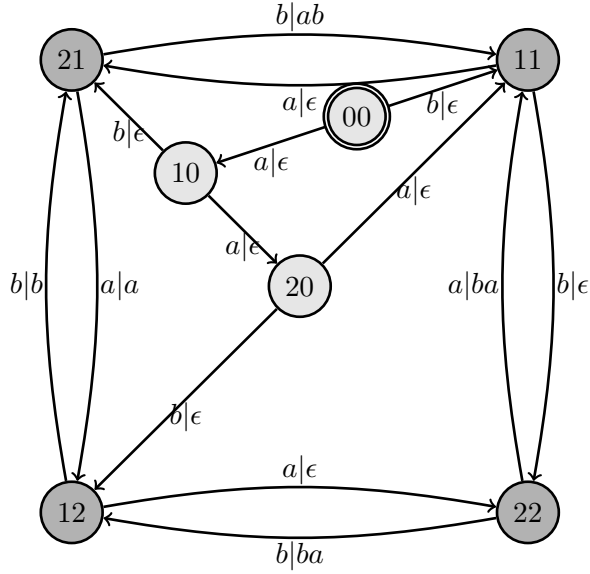

For the lowest interesting value, $D = 3$, the transducer \textswab T can easily be drawn. This diagram is given on figure \ref{fig:transducer}. For readability, we write $a$ (resp. $b$) instead of $0$ (resp. $1$) for the type alphabet, and we omit  the drawing of states which are not connected with the initial one and are not useful for the computation of $t(u)$, for any word $u$.

This transducer has three transient states, ($00$, $10$ and $20$) and four recurrent states ($11, 12, 21$ and $22$) organized in a cycle. A non trivial analysis  of this transducer is given in subsection \ref{ss:D=3}. The  result is stated on the following statement: 

\setcounter{corollary:decrease}{\value{corollary}}
\begin{corollary}\label{corollary:decrease}[$D=3$]
  For any $k$ there exists $n$ in $O(\log k)$ such that for all $u$ of length $k$, $t^n(u)$ is a prefix of  $(ab)^\omega$. 
\end{corollary}

For $D=3$, we therefore establish the emergence of a periodic trace consisting in an alternation of $a$ and $b$. Note that for any $N$, the length of the trace up to $N$ on any interval (where it is defined) is smaller or equal than $N$. Consequently, if we consider the trace up to $N$ on the leftmost interval $I_i$ such that $i(D-1) > \mathcal L(D,N)+3(D-1)$, Corollary \ref{corollary:decrease} states that the trace up to $N$ on $I_{i+O(\log N)}$ is a prefix of $(ab)^\omega$. Since $D$ is fixed equal to 3 and recalling Proposition \ref{lemma:meta2}, we therefore have, for $D=3$, that the trace up to $N$ on $I_{O(\log N)}$ is periodic (figure \ref{fig:fp3}). The next subsection concentrates on the meaning of regular traces.

%%%%%%%%%%%%%%%%%%%%%%%%%%%%%%%%%%
%
%   From traces to waves
%
%%%%%%%%%%%%%%%%%%%%%%%%%%%%%%%%%%

\subsection{From traces to waves}\label{ss:waves}

Corollary \ref{corollary:decrease} states that when the parameter of the model $D$ is equal to $3$, then a periodic trace emerges very quickly (on a logarithmic interval in the number $N$ of grains). The difficulties involved in the generalization of the emergence of regular traces to any parameter $D$ are discussed in the proof of Corollary \ref{corollary:decrease}. Simulations nevertheless let us believe that this behavior generalizes to any parameter, so we choose to present a generalized interpretation of regular traces. 

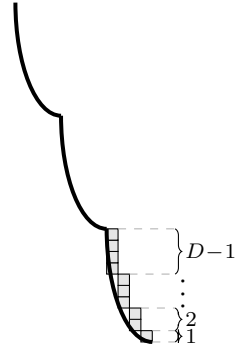
\begin{wrapfigure}{r}{0pt}
  \begin{tikzpicture}[scale=.5]
          \filldraw[fill=black!10] (3*.3,0) rectangle ++ (.3,.3);
          \foreach \y in {0,1}
            \filldraw[fill=black!10] (2*.3,.3+\y*.3) rectangle ++ (.3,.3);
          \foreach \y in {0,...,2}
            \filldraw[fill=black!10] (.3,3*.3+\y*.3) rectangle ++ (.3,.3);
          \foreach \y in {0,...,3}
            \filldraw[fill=black!10] (0,6*.3+\y*.3) rectangle ++ (.3,.3);
          \draw[line width=1.5pt] (4*.3,0) arc (-90:-180:4*.3cm and 10*.3cm);
          \draw[dashed,color=black!30] (4*.3,0) -- ++ (2*.3,0);
          \draw[dashed,color=black!30] (4*.3,.3) -- ++ (2*.3,0);
          \draw[dashed,color=black!30] (3*.3,3*.3) -- ++ (3*.3,0);
          \draw[dashed,color=black!30] (2*.3,6*.3) -- ++ (4*.3,0);
          \draw[dashed,color=black!30] (.3,10*.3) -- ++ (5*.3,0);
          \draw[decorate, decoration=brace] (6*.3,.3) -- node [right] {\scriptsize 1} ++ (0,-.3);
          \draw[decorate, decoration=brace] (6*.3,3*.3) -- node [right] {\scriptsize 2} ++ (0,-2*.3);
          \draw[decorate, decoration=brace] (6*.3,10*.3) -- node [right] {\scriptsize $D\!-\!1$} ++ (0,-4*.3);
          \node at (6*.3+.2,4.5*.3) {\rotatebox{90}{$\dots$}};
          \draw[line width=1.5pt] (0,10*.3) arc (-90:-180:4*.3cm and 10*.3cm);
          \draw[line width=1.5pt] (-4*.3,20*.3) arc (-90:-180:4*.3cm and 10*.3cm);
          %\draw[line width=1.5pt] (-8*.3,30*.3) arc (-90:-180:4*.3cm and 10*.3cm);
        \end{tikzpicture}
  \caption{Regular traces imply regular shape of the fixed point, visibly wavy.}
  \label{fig:wave}
\end{wrapfigure}

\begin{proposition}\label{proposition:wave}
  In KSPM($D$), let $I_i$ be an interval whose columns are greater than $\mathcal L(D,N)+3(D-1)$. Assume that the trace up to $N$ on $I_i$ is
  $$(0,\dots,D-2)^x(0,\dots,p)\text{, with }x\geq 0\text{ and }p \leq D-2\text{.}$$
  Let $y$ be the length of the last $i$-influent-subsequence of type $p$. We have $y \leq x+1$, and $\pi(N)_{[(i+1)(D-1),\infty[}$ equals
  $$\begin{array}[t]\{{ll}. (p,\dots,1)(D\!-\!1,\dots,1)^{x-y}0(D\!-\!1,\dots,1)^y0^\omega &\text{ if }y<x+1\\(p+1,\dots,1)(D\!-\!1,\dots,1)^x0^\omega&\text{ if }y=x+1.\end{array}$$
\end{proposition}

\begin{proof}
  It is a straight induction on avalanches. We concentrate on the right part of fixed points: $\pi(k)_{[(i+1)(D-1),\infty[}$. Initially for $k=0$, it is equal to $0^\omega$. The $D-1$ first $i$-influent-subsequences lead to $D-1,D-2,\dots,1,0^\omega$. From $(D-1,\dots,1)^x0^\omega$, the trace cyclically applies types $0,\dots,D-2$. From Lemma \ref{lemma:similar}, we can predict that each type will be repeated $x+1$ times (there are $x$ peaks to be consumed for each $i$-influent-subsequence), thus each cycle consists of $(D-1)(x+1)$ long avalanches. Each cycle verifies the following invariant: the $((x+1)p + y)^{th}$ long avalanche, $0 \leq p < D-1$ and $0 < y \leq x+1$, has type $p$ and lead to
  \begin{center}
    \begin{tabular}[b]\{{ll}.
      $p,p-1,\dots,1,(D-1,\dots,1)^{x-y}0(D-1,\dots,1)^y0^\omega$ & if $y \leq x$;\\
      $p+1,p,\dots,1,(D-1,\dots,1)^x0^\omega$ & if $y=x+1$.
    \end{tabular}
  \end{center}
  This invariant is proved at each step by a direct application of Theorem \ref{theorem:peak}.\qed
\end{proof}

\begin{remark}
  Note that the application of Theorem \ref{theorem:peak} gives the trace up to $N$ on $I_{i+1}$, it is \mbox{$(0,\dots,D-2)^{x-1}(0,\dots,p)$}. As a consequence, $(0,\dots,D-2)^\omega$ is fixed point for $t$.
\end{remark}

A regular trace for a parameter $D$ is a prefix of $(0,\dots,D-2)^\omega$ and a regular trace on an interval $I_i$ implies a very regular and wavy shape starting from the next interval. The name ``wave'' is explained on figure \ref{fig:wave}.

%%%%%%%%%%%%%%%%%%%%%%%%%%%%%%%%%%
%
%   Analysis of the transducer for D=3
%
%%%%%%%%%%%%%%%%%%%%%%%%%%%%%%%%%%

\subsection{Analysis of the transducer for $D=3$.}\label{ss:D=3}

In this subsection we provide an analysis of the transducer for $D=3$, leading to a proof of Corollary \ref{corollary:decrease}. Note that though we consider maximal length subsequences of long avalanches, input words for the transducer may contain arbitrary numbers of successive occurrences of $a$ and $b$ since we consider only $i$-influent subsequences.

We need some notations. Let $A$ and $A'$ be states of $\mathcal S$ and $u$ be a word of $\mathcal A^*$. Consider, in the transducer, the path which starts in $A$,  whose  sequence of successive edge (left) labels is  given by $u$. We say that we have $A\,  u = A'$ if this  path terminates in $A'$. A word $u$ is an \emph{entry} word if  $00\,  u$ is a recurrent state and for each prefix $u'$ of $u$,  $ 00\, u'$ is a transient state. Let $t_A : \mathcal A^* \to \mathcal A^*$ be the transduction function obtained by changing the initial state by $A$ in the automaton. Hence $t_{00}= t$. We extensively use  $t_{21}$, so we state $t_{21}= t'$. A word $u$ is \emph{basic} for the state $A$ if $\vert t_A(u)\vert  \geq 2$ and for each prefix $u'$ of $u$,  $\vert t_A(u')\vert  < 2$. For each current state $A$, the set of basic words for $A$ and their images by $t_A$ are given below (tables represent case disjunctions according the beginning of $u$)

\begin{center}
  \begin{tabular}[t]{rlcl}
    $(1,1):$ & $aaaa$ & $\rightarrow$ & $aba$\\
    & $aaab$ & $\rightarrow$ & $aba$\\
    & $aab$ & $\rightarrow$ & $ab$\\
    & $ab$ & $\rightarrow$ & $ab$\\
    & $ba$ & $\rightarrow$ & $ba$\\
    & $bb$ & $\rightarrow$ & $ba$
  \end{tabular}
  \begin{tabular}[t]{rlcl}
    $(2,1):$ & $aaa$ & $\rightarrow$ & $aba$\\
    & $aab$ & $\rightarrow$ & $aba$\\
    & $ab$ & $\rightarrow$ & $ab$\\
    & $b$ & $\rightarrow$ & $ab$
  \end{tabular}
  \begin{tabular}[t]{rlcl}
    $(1,2):$ & $aa$ & $\rightarrow$ & $ba$\\
    & $ab$ & $\rightarrow$ & $ba$\\
    & $ba$ & $\rightarrow$ & $ba$\\
    & $bb$ & $\rightarrow$ & $bab$
  \end{tabular}
  \begin{tabular}[t]{rlcl}
    $(2,2):$ & $a$ & $\rightarrow$ & $ba$\\
    & $b$ & $\rightarrow$ & $ba$
  \end{tabular}
\end{center}

Each word $u$ (such that $t(u)  \neq \epsilon$) admits a unique decomposition $u = u_0 u_1 ...u_p$ such that $u_0$ is an entry word, for $1 \leq i < p$,  $u_i$ is a basic word for the state $00\,  u_0 u_1 ...u_{i-1}$, and $u_p$ is a non-empty prefix of a basic word (for the state $00 \, u_0 u_1 ...u_{p-1}$). The word $u$ also admits a decomposition $u =u'_1 u'_2 ...u'_{p'}$ such that for $1 \leq i < p$,  $u'_i$ is a basic word for the state $21\,  u'_0 u'_1 ...u'_{i-1}$, and  $u'_{p'}$ is a non-empty prefix  of a basic word (for the state $21 \, u'_0 u'_1 ...u'_{p'-1}$).

A first result gives us a hint on the form of the sequence of types  applied to successive intervals:

\begin{lemma}\label{lemma:abu}

Let $\mathcal L$ be the language $\mathcal L = \{ab u,  u \in \mathcal A^* \} \cup   \{\epsilon, a\}$.   
\begin{itemize}
\item For each $u \in \mathcal A^*$,  we have  $t'(u) \in \mathcal L$.
\item For each $v \in \mathcal L$,  we have  $t(v) \in\mathcal  L$.
\item For each $u \in \mathcal A^*$ , we have  $t^{2}(u) \in \mathcal L$.
\end{itemize}
\end{lemma}

\begin{proof}
We prove the three items successively, using previous ones as hypothesis. 
\begin{itemize}
  \item Let $u \in \mathcal A^*$ such that $u \neq \epsilon$. Consider the second decomposition seen above: $u =u'_1 u'_2 ...u'_{p'}$. We obtain $t'(u) =  t'(u'_1) t_A(u'_2 ... u'_{p'})$, where $A$ denotes a recurrent state. 
  \begin{itemize}
    \item For $p' \geq 2$,   $t'(u'_1)$ is the image of a basic word for $21$,  thus $t'(u'_1) \in \{ab, aba\}$, which gives $t'(u) \in \mathcal L$. 
    \item For $p' = 1$, $t'(u) =  t'(u'_1)$ and $t'(u'_1)$ is the image of non empty prefix a basic word for $21$, thus $t'(u'_1)$ is a prefix of $aba$, which gives $t'(u) \in \mathcal L$. 
  \end{itemize}
  \item Let $v  \in \mathcal L$. If  $v \in \{\epsilon, a\}$, then $t(v) = \epsilon$. Otherwise $v$ can be written $ab u$. Thus $t(v) = t(ab) t'(u) = t'(u)$, and  $t'(u) \in \mathcal L$ from the first item. This proves: $t(v) \in\mathcal  L$. 
  \item Let $u \in \mathcal A^*$ such that $u \neq \epsilon$.  We consider the first decomposition above:  $u = u_0 u_1 ...u_p$.  We obtain $t(u) =  t_A(u_1) t_{Au_1} ( u_2...u_p)$, where $A$ denotes a recurrent state. 
  \begin{itemize}
    \item For $p = 0$,   $t(u) = \epsilon$,  thus $t^{2}(u) = \epsilon$.  
    \item For $p = 1$,   $t(u) = t_A(u_1)$, and $t_A (u_1)$ is  the image by $t_A$ of a prefix of basic word for $A$,  which gives that $t(u)$ is a prefix of either $aba$ or $ba$  (since possible images of basic words are $ab, ba$, and $aba$). This gives that $t^2(u) \in \{\epsilon, a\}$. 
    \item If $p \geq 2$,  then $t_A(u_1) \in \{ab, ba, aba\}$. If  $t_A(u_1) \in \{ab, aba\}$, then $t(u) \in \mathcal L$, thus $t^2(u) \in \mathcal L$, from the second item. If  $t_A(u_1) = ba$, then we can state $t(u) = ba u'$. Thus $t^2(u) = t'(u')$.  We have $t'(u) \in \mathcal L$ from the first item, thus $t^2(u) \in \mathcal L$. 
  \end{itemize}
\end{itemize}\qed
\end{proof}

\begin{definition}[Height]
  The height $h$ of a finite word $u \in \mathcal A^*$ is $h(u)= \vert |u|_a - |u|_b \vert $ where $|u|_x$ is the number of occurrences of the letter $x$ in $u$.
\end{definition}

\begin{lemma}
  For any finite word $v \in \mathcal L$, we have: $h(t(v)) \leq \frac{h(v)}{4}+1$.
\end{lemma}

\begin{proof}
This is obvious if $v \in \{\epsilon, a\}$. Thus, stating $v = abu$, it remains to prove that, for any finite word $u \in \mathcal A^*$, we have: $h(t'(u)) \leq \frac{h(u)}{4}+1$. 
 
  Let us first consider the case when $|u|_a - |u|_b \geq 0$. Assume that  we  remove a  pattern  of the form $ab$ or $ba$ from $u$. This does not change the value of $h(u)$. Moreover, for each recurrent state $A$, $t_A(ab)$ and  $t_A(ba)$ both are elements of $\{ab, ba\}$ and $A ab = A ba = A$. This  guarantees that  pattern suppression does not change the value of $h(t'(u))$.  
  
  Iterating this argument until there is no more pattern as above leads to the following fact:    if we state $u'=a^{h(u)}$, then we have  $h(t'(u'))=h(t'(u))$.
  
   The integer $h(u)$ can be written as $h(u)=4i+r$, with $r \in \llbracket 0;3 \rrbracket$. We have: $t'(aaaa) = aba$,  and $21 aaaa = 21$. Thus $t'(u') = (aba)^i \,t'(a^r)$, which gives   $h(t'(u')) \leq h((aba)^i)+h(t(r)) \leq i+1$.  Thus    $h(t'(u)) \leq \frac{h(u)}{4}+1$.
 
The other case, when $|u|_a - |u|_b \leq 0$,  is similar. By simplifications of factors $ba$ and $ab$,  we obtain that $h(t'(u'))=h(t'(u))$, for $u'=b^{h(u)}$. The value $h(u)$ can be written as $h(u)=4j+s$, with $s \in \llbracket 0;3 \rrbracket$. We have:  $t'(bbbb) = abbab$ 
and $21 bbbb = 21$. Thus   
   $t'(u') = (abbab)^j \,t'(b^s)$, which gives   $h(t'(u')) \leq h((abbab)^j)+h(t(s)) = j$. Thus  $h(t'(u)) \leq \frac{h(u)}{4}+1$.\qed
\end{proof}

\setcounter{tmp}{\value{corollary}}
\setcounter{corollary}{\value{corollary:decrease}}
\begin{corollary}
  Given a word $u \in \mathcal A^*$ of length $l$, there exists an $n(l)$ in $O(\log{l})$ such that $t^{n(l)}(u)$ is a prefix of $(ab)^\omega$.
\end{corollary}
\setcounter{corollary}{\value{tmp}}

\begin{proof}
We first prove it restricting ourselves on words of $\mathcal L$
  Given a finite word $v$ on $\mathcal L$, we define the maximal height $g(v)= \max \{ | h(v') | ~| v' \text{ prefix of }v\}$. 
  The previous lemma gives the result $g(t(v)) \leq 1 + \frac{g(v)}{4}$. We can now use a trick to get the expected result. We define $g'(v)=g(v)-\frac{4}{3}$, then: 
 $$g(t(v)) \leq 1+\frac{g(v)}{4} \iff g'(t(v)) \leq \frac{g'(v)}{4}$$
 From  lemma  \ref{lemma:abu}, $t(v)$ is element of $\mathcal L$. Thus we can iterate the inequality.  By this way, we obtain,  for each positive integer $n$: 
$$ g'(t^n(v)) \leq \frac{g'(v)}{4^n}$$
Thus, for $n>\log_4 (g'(v)) - \log_4 (\frac{2}{3})$, we have:  $g'(t^n(v)) < \frac{2}{3}$,  so $g(t^n(v)) < 2$ and, by integrity,  
$$g(t^n(v)) \leq 1$$
This last inequality enforces that  $u$ admits a decomposition $t^n(v) = w_1 w_2  ...  w_p$ such that, for $i \in \llbracket 1;p \rrbracket$,   
$w_i \in \{ab, ba\}$, and  $w_p \in  \{\epsilon, a, b\}$. 
Thus $t^{n+1}(u) = t(w_1) t'(w_2)\dots t'(w_p)$.  Thus, $t^{n+1}(u)$ is a prefix of the infinite word $(ab)^\omega$, since $t'(ab) = t'(ba) = ab$ and $t(ab) = t(ba) = \epsilon$. 

Now,  if we take a finite word $u$ on $\mathcal A^*$, we have, from  lemma  \ref{lemma:abu}, $t^2(u) \in \mathcal L$. On the other hand,   
 $\vert t^2(u)   \vert \leq 4 \vert u   \vert$ and $\vert t^2(u)   \vert +\frac{4}{3} \geq g'(t^2(u))$, which gives  $g'(t^2(u)) \leq 4 \vert u   \vert +\frac{4}{3}$. Therefore, 
for $n>\log_4 (4 \vert u   \vert +\frac{4}{3}) - \log_4 (\frac{2}{3})$,  we obtain that  $t^{n+1}(t^2(u))$ is a prefix of the infinite word $(ab)^\omega$. In other words,  for $m  >\log_4{(4\vert u   \vert+\frac{4}{3})}- \log_4 (\frac{2}{3}) +3$, $t^{m}(u)$ is a prefix of the infinite word $(ab)^\omega$.\qed
\end{proof}

For $D=3$, we can now prove the following result:

\setcounter{theorem:D3}{\value{theorem}}
\begin{theorem}\label{theorem:D3}
  For $D=3$ and all $N$, there exists a column $n$ in $O(\log N)$ such that
  $$\pi(N)_{[n,\infty[}=(2,1)^*[0](2,1)^*0^\omega$$
  where $^*$ is the Kleene closure and $[0]$ stands for at most one zero.
\end{theorem}

\begin{proof}
  Let $N$ be given. From proposition \ref{lemma:meta2}, $L(3,N)$ is in $O(\log N)$, therefore there exists an index $m$ in $O(\log N)$ such that traces up to $N$ can be defined on $I_m$, and the transducer for $D=3$ derived from Lemma \ref{lemma:transducer} be applied. Corollary \ref{corollary:decrease} tells that there exists an index $l$ in $O(\log N)$ such that the trace up to $N$ on $I_{m+l}$ is a prefix of $(ab)^\omega$. Finally, Proposition \ref{proposition:wave} gives the result with $n=(m+l+1)(D-1)$. \qed
\end{proof}

%%%%%%%%%%%%%%%%%%%%%%%%%%%%%%%%%%%%%%%%%%%%%%%%%%%%%%%%
%%%%%%%%%%%%%%%%%%%%%%%%%%%%%%%%%%%%%%%%%%%%%%%%%%%%%%%%
%%
%%   Conclusion
%%
%%%%%%%%%%%%%%%%%%%%%%%%%%%%%%%%%%%%%%%%%%%%%%%%%%%%%%%%
%%%%%%%%%%%%%%%%%%%%%%%%%%%%%%%%%%%%%%%%%%%%%%%%%%%%%%%%

\section{Conclusion}

This paper explored emergent regularities in the Kadanoff Sand Pile Model with parameter $D$ in spite of its overall complex behavior. We first studied the sequence of avalanches triggered by the repeated addition of one grain on a stable configuration. Regularities in avalanches allow a precise description of their process by the mean of a distinguished set of columns, peaks. We then concentrate on a particular interval of columns of constant size, $I_i$, and keep track of any firing in this interval, while grain additions and avalanches are repeated. The obtained information, constituted by peaks indices is named a trace. The next step is the construction of a finite state word transducer which computes the trace on $I_{i+1}$ from the trace on $I_i$. Applying again the transducer outputs the trace on $I_{i+2}$, \dots The transducer is thus a finite tool which we can study in order to predict the behavior of avalanches, and hence the shape of fixed points. 

The main result of this paper is:
\setcounter{tmp}{\value{theorem}}
\setcounter{theorem}{\value{theorem:D3}}
\begin{theorem}
  For $D=3$ and all $N$, there exists a column $n$ in $O(\log N)$ such that
  $$\pi(N)_{[n,\infty[}=(2,1)^*[0](2,1)^*0^\omega$$
  where $^*$ is the Kleene closure and $[0]$ stands for at most one zero.
\end{theorem}
\setcounter{theorem}{\value{tmp}}

Computer aided simulations intimate that every result of our study generalizes to any parameter $D$. Nevertheless, some holes remain is this puzzle, summarized on figure \ref{fig:graph}. We therefore conjecture:

\begin{conjecture}\label{conjecture}
  For a general parameter $D$ and all $N$, there exists a column $n$ in $O(\log N)$ such that
  $$\pi(N)_{[n,\infty[}=(D-1,D-2,\dots,2,1)^*[0](D-1,D-2,\dots,2,1)^*0^\omega.$$
\end{conjecture}

Our approach gives a possible way for proving this conjecture. First, prove that avalanche density is 
reached in  $O(\log N)$ columns. Second, understand the structure of the introduced transducer (and its iteration) for a general value $D$.

\begin{figure}[!h]
  \begin{center}
    %\shorthandoff{:;!?} %problem compatibilité frenchb et tikz
  \begin{tikzpicture}[scale=1]
    \draw[->] (0,0) -- ++ (0,4) node[above] {$N$};
    \draw[->] (0,0) -- ++ (8.5,0) node[right] {columns};
    \node at (1.4,1.7) {\textcircled{1}};
    \node at (1.9,1.7) {\textcircled{2}};
    \node at (2.5,1.7) {\textcircled{3}};
    \node at (4,1.7) {\textcircled{4}};
    \filldraw[fill=black!10] (0,1.5) node [left] {\scriptsize $\pi(N)$} rectangle ++ (5.1,-.1);
    \draw[domain=0:8,line width=1pt] plot[id=allure1] function {2**(x/4)-1} node [above] {\scriptsize $\Theta(\sqrt{N})$};
    \draw[domain=0:3.5,line width=1pt] plot[id=allure2] function {2**(x/(3.5/2))-1} node [above] {\scriptsize $\mathcal O(\log N)$};
    \draw[domain=0:2,line width=1pt] plot[id=allure3] function {2**x-1} node [above] {\scriptsize $\mathcal O(\log N)$};
  \end{tikzpicture}
%\shorthandon{:;!?} %problem compatibilité frenchb et tikz
\\
    \begin{tabular}{|l|c|c|}
      \hline
      Results & $D=3$ & For all $D$\\
      \hline
      \textcircled{1}: $\mathcal L(D,N)$ in $O(\log N)$ & Proposition \ref{lemma:meta2} & $\times$\\
      \hline
      \textcircled{2}: Build a transducer & Theorem \ref{theorem:peak} & Theorem \ref{theorem:peak}\\
      & Lemma \ref{lemma:transducer} & Lemma \ref{lemma:transducer}\\
      \hline
      \textcircled{3}: Regular traces & Corollary \ref{corollary:decrease} & $\times$\\
      \hline
      \textcircled{4}: Wave pattern & Proposition \ref{proposition:wave} & Proposition \ref{proposition:wave}\\
      \hline
    \end{tabular}
  \end{center}
  \caption{The precise study of the case $D=3$ gives an asymptotically complete characterization of its fixed points according to the number of grains. Though experimentally confirmed, some parts of this study do not easily generalize to any parameter $D$ of the model. We summarize in this picture our contribution, its results and conjectures.\\
   The left curve labelled \textcircled{1} denotes the global density column $\mathcal L(D,N)$ which has been proved in Proposition \ref{lemma:meta2} to be in $O(\log N)$ for the case $D=3$. From $\mathcal L(D,N)$, we built a finite state word transducer to compute traces from interval to interval in Lemma \ref{lemma:transducer}, using the avalanche process description of Theorem \ref{theorem:peak}. This transducer applies in the part \textcircled{2}, and Corollary \ref{corollary:decrease} states that for $D=3$ the trace up to $N$ is periodic from a column in $O(\log N)$ (label \textcircled{3}). Finally, Proposition \ref{proposition:wave} establishes the relation between regular traces and wave patterns, on part \textcircled{4}.}
  \label{fig:graph}
\end{figure}
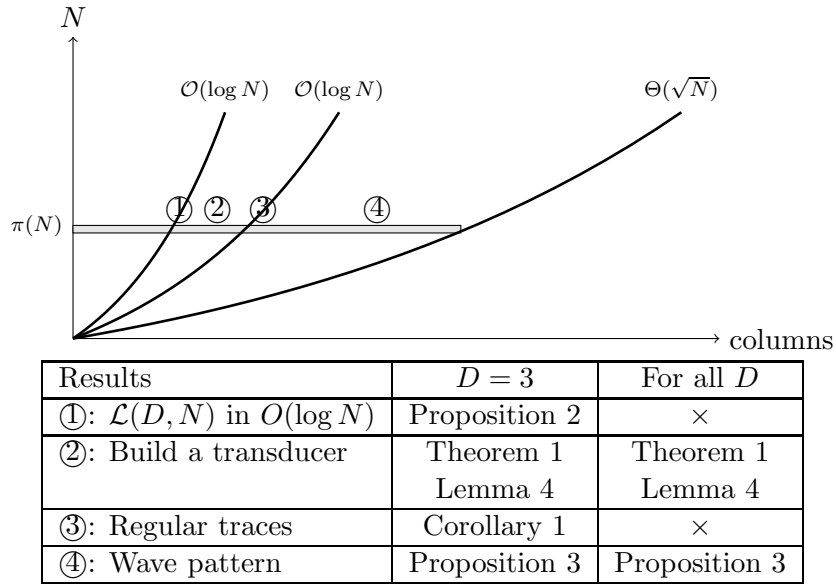

%%%%%%%%%%%%%%%%%%%%%%%%%%%%%%%%%%%%%%%%%%%%%%%%%%%%%%%%
%%%%%%%%%%%%%%%%%%%%%%%%%%%%%%%%%%%%%%%%%%%%%%%%%%%%%%%%
%%
%%   Further work
%%
%%%%%%%%%%%%%%%%%%%%%%%%%%%%%%%%%%%%%%%%%%%%%%%%%%%%%%%%
%%%%%%%%%%%%%%%%%%%%%%%%%%%%%%%%%%%%%%%%%%%%%%%%%%%%%%%%

\section{Further work}

The work presented in this paper has been continued and improved. A different approach, extending the linear algebra analysis introduced in subsection \ref{ss:avD=3} for $D=3$, is presented in \cite{perrot12}. It uses a non-trivial change of basis of the linear system so that its behavior becomes understandable in simple terms, eventually leading to a proof of Conjecture \ref{conjecture}.

\bibliographystyle{elsarticle-harv}
\bibliography{biblio}

\end{document}